\documentclass[11pt]{article}
\usepackage[top=1in,bottom=1in,right=.8in,left=.8in]{geometry}
\usepackage{amssymb,amsmath,amsthm}
\usepackage{tikz}
\usepackage{arydshln}
\usepackage{verbatim}
\usepackage{titlesec}
\usepackage{bbm}
\usepackage{bm}
\usepackage{subfigure}
\usepackage{ stmaryrd }
\usepackage{cleveref}
\usepackage{xcolor}
\usepackage{color}
\usepackage[suppress]{color-edits}
\addauthor{et}{purple}
\addauthor{jg}{blue}
\addauthor{jk}{red}

\titleclass{\subsubsubsection}{straight}[\subsection]

\newcounter{subsubsubsection}[subsubsection]
\renewcommand\thesubsubsubsection{\thesubsubsection.\arabic{subsubsubsection}}

\titleformat{\subsubsubsection}
  {\normalfont\normalsize\bfseries}{\thesubsubsubsection}{1em}{}
  \titlespacing*{\subsubsubsection}
  {0pt}{3.25ex plus 1ex minus .2ex}{1.5ex plus .2ex}

\makeatletter  
\def\toclevel@subsubsubsection{4}
\def\l@subsubsubsection{\@dottedtocline{4}{7em}{4em}}

\makeatother
\crefname{equation}{}{}
\crefrangeformat{equation}{(#3#1#4) --~(#5#2#6)}
\crefmultiformat{equation}{(#2#1#3)}{, (#2#1#3)}{, (#2#1#3)}{, (#2#1#3)}
\crefname{lem}{Lemma}{Lemmas}
\crefname{section}{Section}{Sections}
\crefname{subsubsubsection}{Section}{Sections}
\crefname{rem}{Remark}{Remarks}
\crefname{figure}{Figure}{Figures}
\crefname{table}{Table}{Tables}
\Crefname{lem}{Lemma}{Lemmas}
\crefname{thm}{Theorem}{Theorems}
\Crefname{thm}{Theorem}{Theorems}

\newtheorem{thm}{Theorem}[section]

\newtheorem{example}{Example}[section]
\newtheorem{remark}{Remark}[section]

\newtheorem{lem}{Lemma}[section]

\newtheorem{proposition}[thm]{Proposition}
\newtheorem{corollary}[thm]{Corollary}

\theoremstyle{definition}

\theoremstyle{definition}
\newtheorem{defn}{Definition}[section]
\renewenvironment{proof}{\noindent {\bf Proof.}}{\qed}

\title{Adversarial Perturbations of Opinion Dynamics in Networks}
\author{Jason Gaitonde\footnote{Cornell University, Department of Applied Mathematics. Supported by NSF grant CCF-1408673 and AFOSR grant F5684A1. Email: jsg355@cornell.edu.},\, Jon Kleinberg\footnote{Cornell University, Department of Computer Science. Supported in part by a Simons Investigator Award, a
Vannevar Bush Faculty Fellowship, a MURI grant, a MacArthur Foundation grant, and AFOSR grant F5684A1. Email: kleinberg@cornell.edu.},\, \'Eva Tardos\footnote{Cornell University, Department of Computer Science. Supported in part by NSF grant CCF-1408673,  CCF-1563714, and AFOSR grant F5684A1. Email: eva.tardos@cornell.edu.}}
\date{July 2020}
\begin{document}

\maketitle

\begin{abstract}

In this paper we study the connections between network structure, opinion dynamics, and an adversary's power to artificially induce disagreements. We approach these questions by extending models of opinion formation in the mathematical social sciences to represent scenarios, familiar from recent events, in which external actors have sought to destabilize communities through sophisticated information warfare tactics via fake news and bots. In many instances, the intrinsic goals of these efforts are not necessarily to shift the overall sentiment of the network towards a particular policy, but rather to induce discord. These perturbations will diffuse via opinion dynamics on the underlying network, through mechanisms that have been analyzed and abstracted through work in computer science and the social sciences. 

Here we investigate the properties of such attacks, considering optimal strategies both for the adversary seeking to create disagreement and for the entities tasked with defending the network from attack. By employing spectral techniques, we show that for different formulations of these types of objectives, different regimes of the spectral structure of the network will limit the adversary's capacity to sow discord; in fact, somewhat surprisingly, the \emph{entire} spectrum can be relevant, rather than just the extreme eigenvectors. Via the strong connections between spectral and structural properties of graphs, we are able to qualitatively describe which networks are most vulnerable or resilient against these perturbations. We then consider the algorithmic task of a network defender to mitigate these sorts of adversarial attacks by insulating nodes heterogeneously; we show that, by considering the geometry of this problem, this optimization task can be efficiently solved via convex programming. Finally, we generalize these results to allow for \emph{two} network structures, where the opinion dynamics process and the measurement of disagreement become uncoupled; for instance, this may arise when opinion dynamics are controlled an online community via social media, while disagreement is measured along ``real-world'' connections. We characterize conditions on the relationship between these two graphs that will determine how much power the adversary gains when this occurs. 

\end{abstract}

\section{Introduction}

People's opinions are shaped by their interactions with others.
The resulting process of {\em opinion dynamics} --- the interplay
between opinion formation and the network structure of these
interactions --- has been of great interest in the political science,
sociology, economics, and computer science communities among others.

A range of mathematical models have been proposed for opinion
dynamics in social networks, representing people as nodes in the network.
In typical models, each node in a weighted graph $G=(V,E)$ holds
a real-valued opinion that updates in response to
the values in their network neighborhood. 
In perhaps the most natural model, the
repeated-averaging model of DeGroot, each node repeatedly takes a weighted
average of its own opinion with those of its neighbors; using rich connections between stochastic matrices and Markov chains, it
is possible to show that these dynamics lead to all individuals
sharing a \emph{consensus} opinion in the limit
\cite{degroot1974reaching}. Much work in the economics and social science
literature has studied fundamental questions in this basic model
\cite{golub2010naive}.

A typical feature of real-world networks, however, is that even
approximate consensus is rarely, if ever, achieved. For that reason,
richer models like the Friedkin-Johnsen variation of the
standard DeGroot model \cite{friedkin1990social} or the
bounded-confidence Hegselmann-Krause model
\cite{hegselmann2002opinion} have been proposed to give plausible
explanations for how these long-term disagreements can persist while
maintaining the strong intuition that one will account for and adjust
to the opinions and expressed behaviors of those near them. In this
work, we will take the Friedkin-Johnsen model as a primitive; these
dynamics closely follow the basic DeGroot model, but allow each node
to hold their own immutable internal opinion which always gets
incorporated during the repeated averaging.\\

\textbf{Discord and Adversarial Perturbations.}
These types of mathematical models have been used extensively to study 
both the phenomena associated with consensus and polarization,
as well as the role of interventions designed to shift opinions, 
in which targeting a few individual opinions can shift many others
via the network dynamics.
While the questions have been quite diverse, they share
a crucial underlying assumption, which is that the entities
intervening in the process have a preferred opinion that they
are seeking to promote.

In the real-world settings that motivate these models, however,
a number of important developments have been commanding increasing
attention in the past few years.
One of these is the rising concern about situations in which outside
entities intervene not to promote a specific opinion, but 
instead to induce a {\em lack} of consensus, 
with the explicit intention of
sowing discord. For instance, the U.S. Department of Justice Special
Counsel's Office argues in their 2017 indictment of the Russian
Internet Research Agency (IRA) that the IRA used social media and
targeted advertising with ``a strategic goal to sow discord in the
U.S. political system, including the 2016 U.S. presidential election''
\cite{iraindictment}. This behavior has been seen in multiple countries,
and is not limited to any one actor. 
For example, the doctrine of offensive information warfare by the U.S. 
intelligence community includes provisions for the instigation
of discord between opposing parties \cite{nrc-covert-action}.
More recently, Twitter
disclosed in 2019 that external bots made a coordinated effort to ``sow
political discord in Hong Kong'' with the aim of making protesters
less effective at organizing during the Hong Kong independence
movement \cite{twitter_2019}. It is to be expected that as social
media becomes more and more prevalent, these sorts of external
influences and behaviors will become more frequent.

In such instances, the apparent aim of these attacks is not directly
to promote a specific policy agenda, but rather simply to induce
disagreement; correspondingly, these external actors strive for more
abstract objectives and therefore exhibit more complicated behavior.
For example, analysis of Twitter data in the past few years has
revealed how bots sponsored in some cases by state actors
managed to play both sides of some existing
issue to exploit divisions along racial lines, pro- and anti-vaccine
groups, NFL kneeling, and gun reform, among other matters
\cite{senatereport, broniatowski2018weaponized,diresta2019tactics}.

For these types of scenarios, in which an adversary is seeking
to promote disagreement, we need to look beyond existing 
mathematical formulations of opinion maximization and influence in networks.
We show how to define a new set of questions in the
underlying models of opinion dynamics based on an adversary whose
objective function is some measure of node-to-node disagreement.
In this paper, we formalize such questions; 
roughly speaking, we will start from instances of the 
Friedkin-Johnsen model in which all nodes' internal opinions are at equilibrium, e.g., they are all
equal to $0$, and then we will allow an adversary to perturb
these internal opinion subject to some radius constraint modeling resource constraint of the attacker.
These seeded
internal opinions then diffuse across the 
network to reach an equilibrium under the Friedkin-Johnsen dynamics, 
at which point some measure of disagreement among adjacent nodes is
measured.
Our main goal is to characterize the
interplay between graph structure and the adversary's power: which
networks are most susceptible to or resilient to these kinds of 
perturbations; how can the adversary most efficiently spend 
a bounded amount of resource in order to produce the maximum disagreement;
and how can we most effectively defend a network against such attacks?
We will find that these types of questions bring into play
new network properties that have not typically arisen in 
the analysis of traditional questions around consensus, polarization,
and opinion maximization.

\subsection{Our Results}

\textbf{Spectral Properties of Discord.} In this work
an adversary seeking to maximize 
some measure of {\em discord} among the nodes of a network
would like to perturb the internal opinions of the nodes by
a bounded amount in order to produce as much discord as possible
under the resulting opinion dynamics.

We first show that many such
objectives for such an actor can be naturally posed as the
maximization of an associated quadratic form of the underlying graph
Laplacian. The simple, but crucial observation is that {this
directly connects the adversary's power in attaining their objective
with the eigenvalues of some function of the underlying graph}. The
graph Laplacian is an extremely well-studied object, whose spectral
properties are known to be good approximations of combinatorial
features of the underlying network; this connection will allow us to
qualitatively understand what structural properties of the graph
matter in these different settings.

For different discord objectives, we find that
different regimes of the spectral structure of the underlying networks
characterize the ability of this adversary to succeed in their aim.
In perhaps the most natural case, we define the disagreement 
on an edge of the network to be the squared difference of the opinions
at the edge's endpoints; and we define the disagreement of the network
to be the sum of the disagreements on all its edges.
The challenge for an adversary in maximizing disagreement is that
the same network is being used both by the opinion dynamics to
average away the disagreement, and by the objective function to measure
the disagreement.
We can study this effect in terms of a key one-dimensional parameter
$t$ inherent in the Friedkin-Johnsen opinion dynamics model,
controlling the relative weight of a node's internal opinion and
neighboring opinions in determining how it updates its opinion
in the next time period.
We find, strikingly, that as $t$ ranges from $0$ to $\infty$,
the optimal strategy for the adversary is determined by a sequence
of vectors that ranges over precisely the set of eigenvectors of
the graph Laplacian.
This means that for an adversary seeking to induce maximum disagreement,
the \emph{entire} eigenstructure of the network is relevant.

This dependence on the entire eigenstructure forms a strong contrast
with standard graph objective functions that have a spectral interpretation,
where we typically find a dependence only on the extreme eigenvectors.
Here, by contrast, {\em each} eigenvector controls the optimal choice
for the adversary at some setting of the opinion dynamics.

We also consider additional discord objectives, finding that
the contrast extends here too, with the adversary's optimal behavior
for these other objectives dependent on extreme eigenvectors.
In particular, when the adversary's objective is to maximize
the polarization-disagreement index of Musco, Musco, and Tsourakakis
\cite{musco2018minimizing}, the ability of the adversary is directly
controlled by the second smallest eigenvalue of the graph Laplacian.
Somewhat surprisingly, an infinite-time horizon form of the maximum
disagreement objective is also completely determined by the second
smallest eigenvalue, despite not holding for the one-period time
horizon mentioned above. Through these well-studied connections
between the spectral and combinatorial structures of graphs, we can
then qualitatively describe which networks, under the Friedkin-Johnsen
dynamics, are robust to these sorts of adversarial perturbations.\\

\textbf{Network Defense.} In Section \ref{section:defense}, we
generalize the basic model and consider the converse problem of
\emph{minimizing} the ability of such an adversary to promote
their objective. Following the general principle that the network
topology ought be somewhat difficult to meaningfully change, we allow
a network defender to \emph{insulate} the individuals in the network
heterogeneously. Intuitively, this corresponds to making these
individuals more difficult for an adversary to influence via their
internal opinion. This might be achieved, for instance, through
technological literacy initiatives to make it more difficult for
certain individuals to be susceptible to the normal modes by which an
external actor may seed internal opinions, like bots or fake news.

Formally, we will consider a two-player min-max game, where a network
defender can first heterogeneously \emph{weigh} different individuals
in the network (subject to some explicit normalization); the
interpretation is that the defender can make some individuals less
malleable to seeding. The adversary then solves their objective
function as before, but subject to these new weights. The constraints
on the defender's choice will include the option to keep the
individuals' weights the same as above, so that this optimization
problem for the defender is upper-bounded by the previous analyses. On
the algorithmic side, our main technical contribution in this section
is to show that while the adversary's inner maximization problem no longer has the well-known analytic connection to eigenvalues, for many choices of the adversary's objective and
weight normalization, the network defender's problem can be globally
solved using convex programming. We show this by analyzing the geometry of the defender and adversary's optimization problems. In particular, this implies
the network defender's problem can be solved efficiently.\\

\textbf{Mixed-Graph Objectives.} 
As noted above, the underlying network is playing two distinct
roles in our basic model of disagreement: 
it serves to average opinions of neighboring nodes, and it also 
serves to measure the disagreement between neighbors in the
resulting equilibrium of the opinion dynamics.
But if we look at current discussions of the ways in which
online interaction may be exacerbating conflict among different
opinions, a crucial theme is the way in which the online content
that shapes our opinions is being transmitted across networks that
may be quite different from the social networks linking us to the
people with whom we interact on a daily basis.

These considerations suggest that we return to the adversary's
disagreement maximization
problem, but when the network where disagreement is measured can be different
from the network where the Friedkin-Johnsen opinion dynamics take
place. The individuals and communities one interacts with on social
media can be quite global in nature, as well as relatively anonymous
and self-selecting; in contrast, one's real-world interactions are far
more dictated by other factors like geography and occupation. It is
thus quite natural to think of opinion formation as occurring somewhat
orthogonally to these real-world interactions. To sow discord, an
adversary might hope to induce disagreement measured via these
real-world connections, but must do so filtered through the opinion
dynamics in the online world. Intuitively, if
the opinion graph which does the smoothing of internal opinions and
the disagreement graph look very different, an adversary will be much
more powerful than what is possible in the basic model 
when these graphs must be the same.

To understand this setup, we generalize the adversary's problem with
\emph{mixed-graph} objective functions in the natural way by
considering two distinct graphs at once. One might expect that the
adversary is most limited when the opinion graph and disagreement
graphs are the same, as then the opinion dynamics smooth out the
seeded opinions along the exact same network where disagreements will
be measured. Somewhat surprisingly, we show that it is not always the
case, even subject to some natural degree normalization; there are
explicit, simple disagreement graphs where a different opinion
formation graph provably lowers the adversary's ability to induce
disagreement by a non-negligible amount. However, we also give a
general lower bound that suggests that these examples are somewhat
pathological in that, typically, having different graphs will help the
adversary in this aim compared to the single-graph setting. For other
disagreement graphs, it will in fact turn out that they are optimal
for themselves.

On the one hand, if the measurement and disagreement graphs are highly
similar, edge-by-edge, then the adversary's basic optimization problem
has not really changed from the single-graph setting. This is indeed
the case, via an easy continuity argument, but \emph{not necessary}.
To that end, we show that if the two graphs are good spectral
approximations for each other, then the objective of the adversary is
essentially unchanged. It is well-known using celebrated results of
Batson, Spielman, and Srivastava the every graph, no matter how dense,
has a spectral approximation corresponding to much sparser graph
\cite{batson2012twice}. In particular, this implies that it is not
necessary at all for the two graphs to be \emph{physically} quite
similar, only \emph{spectrally} quite similar.

In the converse direction, it is still quite intuitive that the more
dissimilar the two graphs are, the more an adversary will be able to
induce disagreement. We show that spectral dissimilarity,
appropriately defined, is sufficient for this purpose. Perhaps more
interestingly, by specializing to combinatorial structure, it will
follow that if the two graphs are misaligned in the sense that their
cut structure differs by a large amount along any subset of nodes,
then necessarily the adversary is far more capable of inducing
disagreement than in the single-graph setting. This provides some
theoretical justification for why external actors have been able to
effectively sow discord; this can happen when opinion formation and
disagreement become uncoupled, as is possible with social media platforms.

\subsection{Related Work}

\jgedit{To our knowledge, the focus on interventions to increase discord
in the Friedkin-Johnsen model is new.} In its emphasis on interventions in such models more generally, 
our work has connections to that of
Gionis, et al \cite{gionis2013opinion}. In their work, the authors establish
the $\mathsf{NP}$-hardness of a natural opinion \emph{maximization}
problem work in the Friedkin-Johnsen dynamics. They also derive a
random walk interpretation of these dynamics to establish
submodularity of their problem, enabling a polynomial-time
$(1-1/e)$-approximation algorithm using the natural greedy algorithm.
We focus on different optimization tasks motivated by recent
events aimed at sowing discord in opinion dynamics, 
and moreover, we do not focus on the computational hardness of
such tasks. Indeed, in our setting, the optimization problems will be
computationally tractable. Instead, our emphasis will be on the
relationship between the graph and the adversary's power in achieving
their aims.

Subsequent work by Musco, Musco, and Tsourakakis
\cite{musco2018minimizing} considered a problem, given either a
fixed graph topology or fixed initial opinions, of determining the
optimal choice of the other to minimize their
polarization-disagreement index metric. Their main results shows that
both of these problems can be efficiently solved via convex
optimization. Our approach deviates from theirs in multiple ways: for
one, motivated by these recent events, we are interested in a
\emph{worst-case}, adversarial analysis, not in finding best-case
settings. Moreover, our general philosophy is that the graph topology
for the dynamics is typically quite immutable; the problems we
consider will treat the underlying graphs as fixed. We will also be
principally concerned with analytically characterizing an adversary's
power in these settings in relation to the underlying graph. In a
different direction, prior work by Bindel, Kleinberg, and Oren
interprets the Friedkin-Johnsen dynamics as the Nash equilibrium of a
natural quadratic cost function \cite{bindel2015bad}; they then
characterize the price of anarchy of this equilibrium, namely the
ratio between the cost at equilibrium with that of the global
optimizer. By using spectral techniques, they establish an upper bound
of $9/8$ on this ratio, thereby showing that the equilibrium solution
is a good minimizer of the global cost function.

 \jgedit{Closest to the motivation for our work is an independent paper by Chen and R\'acz \cite{chen2020network} who consider the power of an adversary seeking to induce discord within the Friedkin-Johnsen dynamics, but subject to a sparsity constraint. Their approach relies primarily on carefully bounding the quantities that arise in the associated Laplacian matrices, but without the use of spectral techniques. We explain the connections between our approach and theirs in Section \ref{sec:racz}, and demonstrate how combining our spectral approach with their analysis can sharpen some of the results in their setting.}
 
 From the economics literature, the techniques we use in Section \ref{section:formulation} are similar to those of Galeotti, Golub, and Goyal \cite{galeotti2017targeting}, though with differing motivations. In their work, they consider optimal interventions in the context of \emph{network games} to induce favorable or unfavorable Nash equilibria from the perspective of social welfare, whereas our motivation comes from opinion dynamics. Where these works intersect is via the connection in \cite{bindel2015bad} that the outcome of the Friedkin-Johnsen dynamics can be viewed as the equilibrium behavior of agents in a certain game with quadratic costs; therefore, under certain assumptions on the underlying network and agent utilities, the behavior of agents is equivalent in both settings. Where these approaches then differ is that we treat this equilibrium only as a \emph{behavioral assumption} and thus consider adversaries that attempt to optimize with respect to \emph{different} objectives from the opinion dynamics literature, \emph{not} with respect to social welfare from the game that induces this behavior; \jgedit{for instance, given our differing motivations, we are interested in quantities like inter-agent discord and polarization, which are not equivalent to their social welfare}. This distinction leads to our different conclusions on adversary behavior in Section \ref{section:formulation}, where the entire spectral structure can be exploited by the adversary as opposed to just the extremal parts. We remark that Galeotti, Golub, and Goyal consider optimal interventions in their setting under more general initial conditions via studying the KKT conditions at optimality.

Our current work has thematic elements of all these  papers;
motivated by current events, we consider optimization
tasks in these dynamics. 
As mentioned above, we will largely view the underlying
graph topology as fixed,  and we will consider instead the ability of
an adversary to induce discord by supplying internal opinions
on some issue that previously was at consensus. We will be primarily
concerned with the interplay between graph structure and the
adversary's ability to do so. The opinion dynamics serve to
equilibrate any seeded opinions by an external actor, thus dampening
the effect of the adversary; however, the connections in the graph can
also heighten the resulting disagreement. We also study these
competing effects using primarily spectral techniques and convex
optimization. The network defense problem and mixed-graph objectives
we consider are, to the best of our knowledge, novel to this
literature. These models allow for opinion dynamics and disagreement
to be slightly orthogonal processes and cannot be completely
characterized using the spectral theory of a single graph Laplacian.
We hope that these kinds of generalizations may prove an interesting
direction for future work.

\section{Preliminaries}
\textbf{Notation and Background.}
We start with briefly reviewing some basic facts from spectral graph theory that we will use.
In this work, we will consider simple, undirected, weighted graphs $G=(V,E,w)$, where $\vert V\vert=n$ and $w:E\to \mathbb{R}_{>0}$. We will usually write $m=\sum_{(i,j)\in E} w(i,j)$ for the sum of the weights of all edges in $G$; in the unweighted case, this is just the total number of edges. One can equivalently think of $G$ as being a complete graph, with $w(i,j)=0$ if and only if $(i,j)\not\in E$. We will usually require $G$ to be connected. The adjacency matrix $A\in \mathbb{R}^{n\times n}$ is defined by $A_{i,j}=A_{j,i}=w(i,j)$. Let $D$ be the diagonal degree matrix given by $D_{i,i}=\sum_{j:(i,j)\in E} w(i,j)$ and $0$ off the diagonal. Then the Laplacian matrix of $G$ is given by $L:=D-A$. It is well-known that $L=\sum_{(i,j)\in E} w(i,j)(\mathbf{e}_i-\mathbf{e}_j)(\mathbf{e}_i-\mathbf{e}_j)^T$, where $\mathbf{e}_i\in \mathbb{R}^n$ is the $i$th standard basis vector, and that the quadratic form induced by $L$ is given by
\begin{equation*}
    \mathbf{x}^T L\mathbf{x}=\sum_{(i,j)\in E} w(i,j)(\mathbf{x}(i)-\mathbf{x}(j))^2.
\end{equation*}
It is immediate that $L$ is symmetric and positive semidefinite. In general, $I$ will denote the identity matrix of appropriate dimension.


We use standard notation for the Loewner (positive semidefinite) order, i.e. $M_1\succeq M_2$ if and only if $M_1-M_2\succeq 0$ if and only if $M_1-M_2$ is positive semidefinite. For any connected graph $G$ as above, it is well-known that $L\succeq 0$ and that
\begin{equation*}
    L=\sum_{i=1}^n \lambda_i \mathbf{v}_i\mathbf{v}_i^T,
\end{equation*}
where $0=\lambda_1< \lambda_2\leq \ldots\leq \lambda_n$, the $\mathbf{v}_i$ are orthonormal eigenvectors, and $\mathbf{v}_1=\mathbf{1}/\sqrt{n}$, where $\mathbf{1}$ is the all-ones vector in $\mathbb{R}^n$. We will write $V_i$ for the set of vectors of unit length in the $\lambda_i$-eigenspace of $L$. Observe that if $\lambda_i\neq \lambda_j$ for all $i\neq j$, then $V_i=\{\pm \mathbf{v}_i\}$. For more on the spectral theory of graphs, see for instance \cite{chung1997spectral, godsil2013algebraic}.

Given a symmetric matrix $X\in \mathbb{R}^{n\times n}$ with eigendecomposition as above (though with not necessarily nonnegative eigenvalues) and function $f:\mathbb{R}\to \mathbb{R}$, we define
\begin{equation*}
    f(X)=\sum_{i=1}^n f(\lambda_i) \mathbf{v}_i\mathbf{v}_i^T.
\end{equation*}
Note that if we stipulate that $f(y)\geq 0$ for $y\geq 0$, then if $X\succeq 0$, $f(X)\succeq 0$. For $X\succeq 0$, we write $\|X\|$ for the operator norm, or equivalently, the largest eigenvalue.

We will write $\|\mathbf{x}\|_p=\big(\sum_{i=1}^n \vert\mathbf{x}(i)\vert^p\big)^{1/p}$ for the usual $\ell_p$ norms with $1\leq p\leq \infty$, and $\|\mathbf{x}\|_0=\vert \{i\in [n]:x_i\neq 0\}\vert$ for the sparsity-``norm'' (note that this is not a norm in the usual sense). We will be interested in this paper in the setting of $p=0,1,2,\infty$. We will also consider \emph{weighted}-$\ell_2$ norms, defined for a vector $\mathbf{w}\in \mathbb{R}^n_{>0}$ by $\|\mathbf{x}\|_{\mathbf{w}}=\sqrt{\sum_{i=1}^n \mathbf{w}(i)\mathbf{x}(i)^2}$. Note that with this convention, $\|\cdot\|_2=\|\cdot\|_{\mathbf{1}}$.\\

\noindent\textbf{Friedkin-Johnson Dynamics.}
In this paper we will assume that opinions evolve using the Friedkin-Johnson dynamics \cite{friedkin1990social} (FJ dynamics), which we describe next. The dynamic
is specified by an undirected simple graph $G=(V,E,w)$ and an initial, internal opinion vector $\mathbf{s}\in \mathbb{R}^n$. Starting with $\mathbf{z}^{(0)}=\mathbf{s}$, each node $i\in V$ updates her opinion by taking the weighted average of her neighbors in $G$, \emph{as well her own internal opinion}:
\begin{equation*}
    \mathbf{z}^{(t+1)}(i)=\frac{\mathbf{s}(i)+\sum_{j:(i,j)\in E} w(i,j) \mathbf{z}^{(t)}(j)}{1+\sum_{j:(i,j)\in E} w(i,j)}.
\end{equation*}
Notice that these equations implicitly normalize a weight of $1$ on one's private opinion. It is well-known that these dynamics converge to a fixed point, and the limiting final opinion vector is given by
\begin{equation*}
    \mathbf{z}=(I+L)^{-1}\mathbf{s},
\end{equation*}
where $L$ is the Laplacian of $G$.

\section{Adversarial Optimization with Friedkin-Johnsen Dynamics}
\label{section:formulation}
Throughout this section,
we will consider variations on the following question: what initial opinions should an adversary induce in a network maximize some objective function after the opinions diffused
according to the Friedkin-Johnsen dynamics?
And moreover, what qualitative features of the graph make it robust to these perturbations? 

We will show how many interesting and natural adversarial objectives can be posed as \emph{quadratic forms} related to the underlying graph Laplacian $L$. Moreover, we show how different interesting regimes of the spectrum and eigendecomposition of $L$ dictate the adversary's power to achieve their objective for different variants. Formally, we will typically (though not always) be concerned in various formulations of the following optimization problem for an adversary:
\begin{equation}
\label{eq:advproblem}
    \max_{\mathbf{s}\in \mathbb{R}^n: \|\mathbf{s}\|_2\leq R} \mathbf{s}^T(I+L)^{-1} f(L) (I+L)^{-1}\mathbf{s},
\end{equation}
where $f:\mathbb{R}\to \mathbb{R}$ satisfies $f(y)\geq 0$ for $y\geq 0$; this restriction is made to ensure that the above quadratic form is nonnegative.\footnote{\jgedit{We discuss different forms of budget constraints later in this section. We shall see that spectral analysis can be used to study such constraints, though the connection to graph structure is not as clear.}}

The interpretation is that, on some issue that is currently at consensus in the graph, the adversary will first supply initial opinions $\mathbf{s}$, for instance via fake news or targeted advertisements. This is done subject to a fixed norm constraint, which corresponds to a limited budget to perturb initial opinions in this way. These opinions then diffuse and become ``smoothed'' through the underlying graph $G$ via the Friedkin-Johnsen dynamics. The goal of the adversary is to choose these initial opinions in order to maximally induce some desired effect, \emph{knowing} that whatever initial opinions it seeds, the opinion dynamics dictated by the underlying graph will inevitably partially equilibrate them. Of course, we will be interested in functions $f$ where the above holds some sort of physical meaning.

Two of the goals we consider for the adversary are the notions of disgareement and polarization introduced by
Musco, Musco, and Tsourakakis \cite{musco2018minimizing}:
\begin{defn}
Given any opinion vector $\mathbf{x}\in \mathbb{R}^n$ and undirected graph $G$, the \emph{disagreement} of $G$ with opinions $\mathbf{x}$ as
\begin{equation*}
    \sum_{(i,j)\in E} w(i,j) (\mathbf{x}(i)-\mathbf{x}(j))^2=\mathbf{x}^T L\mathbf{x},
\end{equation*}
where $L$ is the graph Laplacian of $G$.

Similarly, given just an opinion vector $\mathbf{x}$, we write 
\begin{equation*}
    \overline{\mathbf{x}}:=\mathbf{x}-\frac{\mathbf{x}^T \mathbf{1}}{n}\cdot \mathbf{1}= \big(I-\frac{1}{n}\mathbf{1}\mathbf{1}^T\big)\mathbf{x}
\end{equation*}
as the \emph{de-meaned} version of $\mathbf{x}$ obtained by subtracting off the average of $\mathbf{x}$ from each component. Then the \emph{polarization} of $\mathbf{x}$ is
\begin{equation*}
    \|\overline{\mathbf{x}}\|_2^2=\sum_{i=1}^n \overline{\mathbf{x}}(i)^2.
\end{equation*}
In words, the polarization is a measure of variance for $\mathbf{x}$.
\end{defn}

\subsection{Disagreement}
\label{section:disagreement}
First, we consider an adversary that seeds initial opinions subject to a norm constraint with the goal of \emph{maximizing disagreement}. In the framework above, this can be realized via the choice $f(y)=y$, and yields the simple objective function
\begin{equation}
\label{eq:disagreement}
    \max_{\mathbf{s}\in \mathbb{R}^n: \|\mathbf{s}\|_2\leq R}\mathbf{s}^T(I+L)^{-1}L(I+L)^{-1}\mathbf{s}.
\end{equation}

The crucial feature of this objective is the fact that $L$ both dictates the measurement of disagreement, and the opinion dynamics themselves. Somewhat surprisingly, for this objective function, the \emph{entire} eigenstructure of the graph can matter; as $L$ is scaled in importance compared to $I$, the adversary passes through each eigenspace in order as they optimize.
\begin{thm}
\label{thm:dis}
For any graph $G$, the objective function (\ref{eq:disagreement}) is upper bounded by $\frac{R^2}{4}$. Moreover, if $G$ is connected and we consider the family of problems (\ref{eq:disagreement}) with $L(t)=tL$ for $t>0$, \jgedit{then for any $i> 1$, there exists a $t>0$ such that $R\cdot V_i$ is an optimizer for this value of $t$.}
\end{thm}
\begin{proof}
By the variational characterization of eigenvalues of symmetric matrices, this amounts to understanding the spectrum of the matrix $(I+L)^{-1}L(I+L)^{-1}$. This matrix has eigenvalues 
\begin{equation*}
    \frac{\lambda_i(L)}{(1+\lambda_i(L))^2}, \quad i=1,\ldots,n,
\end{equation*}
with the same corresponding eigenvectors as $L$. In particular, the optimal value of (\ref{eq:disagreement}) turns out to be $R^2\cdot \bigg(\max_{i\in [n]} \frac{\lambda_i(L)}{(1+\lambda_i(L))^2}\bigg)$. Consider now the function $g(x)=\frac{x}{(1+x)^2}$. It is easy to compute
$g'(x) = \frac{1-x}{(1+x)^3}$,
from which it immediately follows that $g$ is increasing for $0\leq x\leq 1$ and decreasing for $x>1$, attaining a peak of $1/4$ at $x=1$. This immediately gives the upper bound, where equality holds if and only if $\lambda_i(L)=1$ for some $i\in [n]$.

For the second statement, as we vary $L(t)=t\cdot L$ with $t>0$, the eigenvectors of $L$ remain the same, but the eigenvalues scale as $t\cdot \lambda_2,\ldots, t\cdot \lambda_n$. As $t$ varies strictly between $0$ to infinity, each nonzero eigenvalue of $t\cdot L$ passes through the peak of the function $g(x)$ at $x = 1$ when $t=1/\lambda_i(L)$. This implies that for this value of $t$, the set of optimizers is $R\cdot V_i$, the length $R$ vectors in $V_i$, the $\lambda_i$-eigenspace.\footnote{We remark that for certain values of $t$, there may exist optimizers that are \emph{not} eigenvectors of $L$. The reason is that because $g$ is not injective, there may exist distinct eigenvalues of $L$ that are maximal for $g$, so that both eigenspaces are optimal, and therefore the span of these eigenspaces are optimal as well, even though most such vectors are not eigenvectors.}
\end{proof}\\

One way to interpret the previous theorem is by first considering the extreme ranges of $tL$. For $t\approx 0$, $tL$ is negligible compared to $I$, and so $(I+tL)^{-1}\approx I$. Physically, this corresponds to each individual not really listening to their neighbors over their own initial opinion, so the opinion dynamics do not substantially change the seeded opinion by the adversary. When this is the case, the above shows that the optimal strategy for the adversary that seeks to induce maximum disagreement is to seed vectors in the direction of $V_n$. In this case, by the quadratic form (\ref{eq:disagreement}), the adversary's strategy is to simply feed in opinions that directly maximize disagreement in $G$, as the opinion dynamics are themselves negligible. The actual vectors in the set $V_n$ attempt to place different values on the two sides of each edge; in that sense, the vectors in $V_n$ can be thought of as solving a type of graph coloring problem \cite{alon1997spectral}. The quantity $\lambda_n(L)$ itself is also known to be closely related to the size of maximum cuts in graphs, see for instance \cite{delorme1993laplacian, goemans1995improved, trevisan2012max}.

On the other extreme, for large $t$, $(I+tL)^{-1}(tL)(I+tL)^{-1}\approx \frac{1}{t} L^{\dagger}$, where $L^{\dagger}$ is the pseudo-inverse of $L$. In this regime, the largest eigenvalue of $L^{\dagger}$ is the inverse of the smallest nonzero eigenvalue of $L$, which is just $\lambda_2$, with corresponding optimizer proportional to $\mathbf{v}_2$. In general, it is well-known that $\lambda_2$ and $\mathbf{v}_2$ are intimately connected to the \emph{normalized sparsest cut} of $G$; if $G$ is $d$-regular, then the famous discrete Cheeger inequality asserts that $\lambda_2$ is an approximation to the normalized sparsest cut to a factor of $\Theta(\sqrt{d})$ \cite{chung1997spectral}. It is also known that some sweep cut of $\mathbf{v}_2$ yields a bipartition that attains this bound. Therefore, in this large $t$ regime where graph neighbors are higher weighted than internal opinions, the initial opinion vector inducing maximal disagreement roughly corresponds to a sparse cut. Because the graph interactions in the Friedkin-Johnsen dynamics are so strong compared to the weight of the internal opinions, the optimal strategy of the adversary is roughly to induce disagreement on along some sparse cut of the graph.

As $t$ varies between these two regimes, the above result implies that \emph{every} nontrivial eigenvector becomes relevant. In the intermediate range of $t$, the relative effects of $tL$ in measuring disagreement and $tL$ in smoothing the initial opinions via the opinion dynamics directly conflict, which causes the adversary to pass through each eigenspace. For different regimes of how the each individual weighs their internal opinion to those of their neighbors, these effects balance in different ways, leading to this more interesting connection between the eigenstructure of $L$ and the adversary's ability to induce large disagreement in a network with these dynamics.

\subsection{Repeated Disagreement}
As a simple extension of one-period disagreement considered in Section \ref{section:disagreement}, it is natural to consider a similar objective taken over a longer time horizon. An analogous multi-period setting is the following: the adversary supplies the initial opinions $\mathbf{s}$, which the Friedkin-Johnsen dynamics take to $(I+L)^{-1}\mathbf{s}$. In the first period, the disagreement of this equilibrium opinion vector is measured, as before. In the next period, the last period final opinions $(I+L)^{-1}\mathbf{s}$ become the new initial opinions, which subsequently get updated by the dynamics to $(I+L)^{-2}\mathbf{s}$; the disagreement of these opinions is then measured and added to the first disagreement. This process is repeated for $T+1$ periods, where $T\in \mathbb{N}\cup \{\infty\}$. 

\jgedit{One natural setting for this objective is the following: consider a setting where opinions solidify, then update on a cyclic basis. For instance, during a U.S. election cycle, opinions might form along network structure according to the FJ dynamics; however, once the elections happen, the final opinions can be viewed as priors, or innate opinions, for the next cycle, at which point the dynamics will take hold once again and so on. These sorts of repeated dynamics are consider by Chitra and Musco \cite{10.1145/3336191.3371825} in their work on understanding filter bubbles in the FJ dynamics.} 

In this multi-period setting, the adversary's problem is to maximize the total disagreement across all time periods: putting this in our framework, this is
\begin{equation*}
    \max_{\mathbf{s}\in \mathbb{R}^n:\|\mathbf{s}\|_2\leq R} \mathbf{s}^T(I+L)^{-1}\bigg(L+(I+L)^{-1}L(I+L)^{-1}+\ldots+(I+L)^{-T}L(I+L)^{-T}\bigg)(I+L)^{-1}\mathbf{s}.
\end{equation*}
To understand this setting, notice that we may work orthogonal to $\mathbf{1}$, as this gives zero objective value. In this case, the relevant $f$ function that acts on the eigenvalues becomes
\begin{equation*}
    f(x)=\sum_{i=0}^{T} \frac{x}{(1+x)^{2i}}=\frac{-(1+x)^{-2T}+(1+x)^2}{2+x},
\end{equation*}
with the convention that the first term is zero when $T=\infty$. We thus find that 
\begin{align*}
    \max_{\mathbf{s}\in \mathbb{R}^n:\|\mathbf{s}\|_2\leq R} \mathbf{s}^T(I+L)^{-1}\bigg(L+(I+L)^{-1}L(I+L)^{-1}+\ldots&+(I+L)^{-T}L(I+L)^{-T}\bigg)(I+L)^{-1}\mathbf{s}\\
    &=R^2\cdot \max_{1<i\leq n}\bigg\{\frac{-(1+\lambda_i(L))^{-2(T+1)}+1}{2+\lambda_i(L)}\bigg\}
\end{align*}
Note that for $T=0$, this reduces to the analysis of the first subsection. 
Rather curiously, for $T=\infty$, this implies that the optimizer of this infinite-period game is precisely $R\cdot V_2$, with corresponding objective value $\frac{R^2}{2+\lambda_2(L)}$. While in the one-shot game considered above, any one of the eigenvectors could be the relevant optimizer, the optimizer of this problem necessarily lies in the direction of $V_2$.

\subsection{Polarization and Disagreement}
For a different objective, Musco, Musco, and Tsourakakis \cite{musco2018minimizing} considered the following cost metric on opinions, which they term the \emph{polarization-disagreement index}, obtained by taking the sum of the disagreement of $\mathbf{x}$ in $G$ and the polarization of $\mathbf{x}$; the authors show that when this measure is done with final opinions $\mathbf{x}=(I+L)^{-1}\mathbf{s}$,  the polarization-disagreement index of the Friedkin-Johnsen dynamics with graph $G$ and initial opinions $\mathbf{s}$ can be simplified to
\begin{equation*}
    \overline{\mathbf{s}}^T(I+L)^{-1}\overline{\mathbf{s}}.
\end{equation*}
Note that this is not quite of the form (\ref{eq:advproblem}), though it is quite similar.

Suppose that an external adversary now chooses the internal opinions $\mathbf{s}$ to maximize the polarization-disagreement index of the final opinions after undergoing the Friedkin-Johnsen dynamics:
\begin{equation*}
    \max_{\mathbf{s}\in \mathbb{R}^n: \|\mathbf{s}\|_2\leq R}\overline{\mathbf{s}}^T(I+L)^{-1}\overline{\mathbf{s}}= \max_{\mathbf{s}\in \mathbb{R}^n: \|\mathbf{s}\|_2\leq R} \mathbf{s}^T(I-\frac{1}{n}\mathbf{1}\mathbf{1}^T)(I+L)^{-1}(I-\frac{1}{n}\mathbf{1}\mathbf{1}^T)\mathbf{s}.
\end{equation*}

\begin{thm}
\label{thm:polarizationdisagreement}
For any graph $G$, the maximizer of the above maximization problem is $\pm R\cdot V_2$, and the resulting objective value is $\frac{R^2}{1+\lambda_2}$.
\end{thm}
\begin{proof}
First, notice that for any vector $\mathbf{s}$, we have
$\|\mathbf{s}\|_2\geq \|\overline{\mathbf{s}}\|_2$;
this follows from the Pythagorean theorem and decomposing $\mathbf{s}$ into the parts orthogonal to $\mathbf{1}$ (namely, $\overline{\mathbf{s}}$) and the projection onto $\mathbf{1}$. Because the above maximization function only depends on the de-meaned version of $\mathbf{s}$ and is homogenous in $\mathbf{s}$, we may assume the adversary restricts to the subspace orthogonal to $\mathbf{1}$. As such, the problem becomes
\begin{equation*}
    \max_{\mathbf{s}\in \mathbb{R}^n: \|\mathbf{s}\|_2\leq R, \mathbf{s}\perp \mathbf{1}}\mathbf{s}^T(I+L)^{-1}\mathbf{s}.
\end{equation*}
The eigenvalues of $(I+L)^{-1}$ are $\frac{1}{1+\lambda_i(L)}$ for $\lambda_i(L)$ the eigenvalues of $L$; as $\mathbf{1}$ is the eigenvector for the largest eigenvalue $1$ of this matrix, the variational characterization of eigenvalues implies that the above is exactly
\begin{equation*}
    R^2\cdot \lambda_{n-1}((I+L)^{-1})=\frac{R^2}{1+\lambda_2(L)},
\end{equation*}
and this is attained by the set of vectors $R\cdot V_2$, as desired.
\end{proof}\\

In particular, it follows that when an adversary chooses $\mathbf{s}$ to maximize the polarization-disagreement index of $G$ under these dynamics, the only relevant structure of the network that determines its robustness to these adversarial perturbations is precisely determined by the second smallest eigenvalue, and the initial opinion vector inducing this are the corresponding eigenvectors. Under the connection between the second smallest eigenvalue and eigenvectors of $L$ and sparse cuts discussed above, Theorem \ref{thm:polarizationdisagreement} essentially asserts that the ability of an adversary to induce polarization-disagreement in a graph $G$ is essentially determined by the existence of small normalized cuts. Moreover, the actual optimizer for initial opinion vector roughly places large values on one side of a small cut and smaller values on the other side. As an immediate corollary of the above result, we can obtain the following intuitive, but nontrivial fact:
\begin{corollary}
\label{cor:completeoptimal}
Let $\mathcal{L}$ be the set of Laplacians of weighted $n$-node graphs subject to the total edge weight normalization $\text{Tr}(L)=2m$. Then
\begin{equation*}
    \arg\min_{L\in \mathcal{L}}\bigg\{\max_{\mathbf{s}\in \mathbb{R}^n: \|\mathbf{s}\|_2\leq R}\overline{\mathbf{s}}^T(I+L)^{-1}\overline{\mathbf{s}}\bigg\}=\frac{2m}{n(n-1)}\cdot  L_{K_n},
\end{equation*}
where $L_{K_n}$ is the Laplacian of the unweighted simple complete graph $K_n$.\footnote{The factor $\frac{2m}{n(n-1)}$ is just to satisfy the total edge weight condition.}
\end{corollary}

\begin{proof}
From the proof of Theorem \ref{thm:polarizationdisagreement}, for any fixed Laplacian $L$, the inner maximization yields the objective value $\frac{R^2}{1+\lambda_2(L)}$; therefore, the claim is equivalent to showing that the Laplacian $\frac{2m}{n(n-1)}\cdot  L_{K_n}$ has the maximum second-smallest eigenvalue among all $L\in \mathcal{L}$.

To see this, observe that for any such $L$, $\lambda_1(L)=0$, and therefore
\begin{equation}
\label{eq:lambda2lb}
    2m = \text{Tr}(L) = \sum_{i=2}^n \lambda_i(L)\geq (n-1)\cdot \lambda_2(L).
\end{equation}
It immediately follows that for any such Laplacian $L$, $\lambda_2(L)\leq \frac{2m}{n-1}$; we will be done if we can show that this is attained for the claimed weighted complete graph. But it is not difficult to check that 
\begin{equation}
    \lambda_2(L_{K_n})=n;
\end{equation}
after scaling by $\frac{2m}{n(n-1)}$ so that the trace condition is satisfied, we see that this upper bound on $\lambda_2$ is exactly attained, proving optimality. Moreover, this occurs if and only if all of the nonzero eigenvalues are $\frac{2m}{n-1}$ by virtue of (\ref{eq:lambda2lb}), which occurs if and only if the graph is the scaled complete graph.
\end{proof}

This corollary thus states that the complete graph, appropriately weighted, is min-max optimal given the adversary's objective to induce maximal polarization-disagreement index when running the Friedkin-Johnsen dynamics. More generally, any \emph{spectral expander} will be robust to these adversarial perturbations, see for instance \cite{hoory2006expander}.

\subsection{Absolute Displacement}
Along these lines, suppose now that the attacker simply seeks to displace opinions maximally from the consensus at $0$, measured in Euclidean norm. This too can be realized in the above setting: suppose that $f(y)\equiv 1$, so that $f(L)=I$. In this case, the adversary solves the following problem:
\begin{equation*}
    \max_{\mathbf{s}\in \mathbb{R}^n: \|\mathbf{s}\|_2\leq R} \mathbf{s}^T(I+L)^{-2}\mathbf{s}=\max_{\mathbf{s}\in \mathbb{R}^n: \|\mathbf{s}\|_2\leq R} \|(I+L)^{-1}\mathbf{s}\|^2_2.
\end{equation*}
This latter identity shows that for this choice of $f$, the adversary's goal is indeed to maximize the $\ell_2$-norm of the final opinion vector $(I+L)^{-1}\mathbf{s}$, or equivalently, to \emph{displace} the final opinions from the initial consensus at $\mathbf{0}$ as much as possible. However, 
\begin{equation*}
    (I+L)^{-2} = \sum_{i=1}^n \frac{1}{(1+\lambda_i)^2} \mathbf{v}_i\mathbf{v}_i^T.
\end{equation*}
As before, we thus obtain
\begin{equation*}
    \max_{\mathbf{s}\in \mathbb{R}^n: \|\mathbf{s}\|_2\leq R} \mathbf{s}^T(I+L)^{-2}\mathbf{s}=R^2 \cdot \lambda_{n}((I+L)^{-2}).
\end{equation*}
But recall that for any graph, the smallest eigenvalue of $L$ is $0$, with corresponding eigenvector $\pm\frac{1}{\sqrt{n}}\cdot\mathbf{1}$ (and this is unique if $G$ is connected); as a result,
\begin{equation*}
    \lambda_{\text{max}}((I+L)^{-2})= 1,
\end{equation*}
and the unique maximizer (up to sign) is $\mathbf{s}=\frac{R}{\sqrt{n}}\mathbf{1}$. In particular, for this optimization problem, the network topology plays no role at all. This observation is quite similar to one made in the context of opinion maximization in \cite{gionis2013opinion}.

\subsection{Other Adversary Constraints}
\jgedit{In the previous sections, we focused on the $\ell_2$-budget constraint, as this most cleanly elucidates the nontrivial interplay between spectral structure and the optimal choices of the adversary for given objectives of discord. In this section, we provide positive results in other settings; the first comes from the independent work of Chen and R\'acz \cite{chen2020network}, who consider the same disagreement objective as we do, but instead makes the adversary sparsity-constrained, as well as restricted to a cube. We first show how a combination of the spectral ideas above with their analysis can lead to a more refined bound on the adversary's power with these constraints.}

\jgedit{We then consider the problem of the adversary choosing weights subject to an $\ell_{\infty}$ constraint, say $\|s\|_{\infty}=1$; it is not difficult to show that the adversary will choose a seed such that each component has absolute value $1$. One can think of this setting as allowing the adversary to partition the agents, and then seed them separately a bounded amount. In this case, one can of course give a trivial spectral bound for each objective. However, using known techniques from combinatorial optimization, we show that it is possible to efficiently obtain an $O(1)$ approximation to the adversary's power in this setting, and moreover, using known algorithms in this setting, obtain an explicit near-optimal vector for the adversary.}

\jgedit{Finally, we consider the problem of an $\ell_1$ constrained adversary. However, in this version of the problem, the problem becomes somewhat trivial, in that the adversary must pick from a fixed set of vectors that themselves are graph-independent; still, we give a spectral argument that nonetheless gives some bound on the adversary's power, that is provably tight for the class of vertex-transitive graphs.}

\subsubsection{Sparsity-Constrained Adversary: Connection to Chen-R\'acz}
\label{sec:racz}
We now explain the setting independently considered by Chen and R\'acz \cite{chen2020network}, as well as show how the combining the techniques considered here with their methods can provide a slight strengthening of their theoretical results. Their setup is the following: first, they allow the network to have nonzero initial starting opinions $\mathbf{s}_0\in [0,1]^n$ and allow the adversary to perturb the initial opinions in $[0,1]^n$, but subject to a $k$-sparsity constraint. That is, in their model, they consider the following optimization problem for the adversary:
\begin{equation}
    \max_{\mathbf{s}\in [0,1]^n: \|\mathbf{s}-\mathbf{s}_0\|_0\leq k} \mathbf{s}^T(I+L)^{-1}L(I+L)^{-1}\mathbf{s}.
\end{equation}

We now show how to combine their analysis with our spectral techniques to obtain slightly sharper bounds than the results in their paper on the amount an adversary can induce discord.
\begin{lem}[Lemma 4.1 of \cite{chen2020network}]
\label{lem:racz1}
For any $\mathbf{s}$ satisfying the above constraints,
\begin{equation}
    \|(I+L)^{-1}(\mathbf{s}-\mathbf{s}_0)\|_1\leq k.
\end{equation}
\end{lem}
\begin{proof}
It is well-known that matrix norm induced by the $\ell_1$ vector norm is exactly the maximum largest column sum. As $\mathbf{1}$ is a right and left eigenvector of $(I+L)^{-1}$ with eigenvalue $1$, and $(I+L)^{-1}_{ij}=\mathbf{e}_i^T(I+L)^{-1}\mathbf{e}_j\geq 0$ by the repeated-averaging interpretation of the dynamics, this implies that $I+L$ is doubly stochastic, so every absolute column sum is $1$. Thus $\|(I+L)^{-1}(\mathbf{s}-\mathbf{s}_0)\|_1\leq \|\mathbf{s}-\mathbf{s}_0\|_1\leq \|\mathbf{s}-\mathbf{s}_0\|_0\leq k$, where the last fact uses the fact the difference componentwise is at most $1$ and the sparsity constraint.

\end{proof}

\begin{lem}
\label{lem:racz2}
For $\mathbf{s}_0\in [0,1]^n$,
\begin{equation*}
    \|L(I+L)^{-1}\mathbf{s}_0\|_{\infty}\leq d_{max},
\end{equation*}
where $d_{max}$ is the largest degree in $G$.
\end{lem}
\begin{proof}
Note that $\mathbf{z}_0:=(I+L)^{-1}\mathbf{s}_0\in [0,1]^n$ as we have seen that $(I+L)^{-1}$ is doubly stochastic, so is a positive contraction in $\ell_{\infty}$ (this also follows from the repeated-averaging interpretation). Then
\begin{equation*}
    \vert (L\mathbf{z}_0)_i\vert=\bigg\vert \sum_{j=1}^n L_{i,j}(\mathbf{z}_0)_j\bigg\vert= \bigg\vert d_{i}(\mathbf{z}_0)_i-\sum_{j\neq i} A_{ij} (\mathbf{z}_0)_j\bigg\vert \leq \max\{d_{i}(\mathbf{z}_0)_i,d_i\max_{j}(\mathbf{z}_0)_j\}\leq d_{\max}.
\end{equation*}
\end{proof}

From these lemmas, we can now combine the ideas in \cite{chen2020network} with the spectral approach of this work to obtain a sharpened bound in this setting when $k=\Omega(\sqrt{n})$. For convenience, we will write $(I+L)^{-1}L(I+L)^{-1}$ as $\Sigma$.
\begin{thm}[Theorem 1.3 of \cite{chen2020network}]
For any graph $G$, the amount of disagreement an adversary can add is bounded by
\begin{equation}
    \lambda_{max}(\Sigma)k+\sqrt{k}\min\{2 d_{max}\sqrt{k},2\cdot \lambda_{max}(\Sigma)\sqrt{n}\}.
\end{equation}
\end{thm}
\begin{proof}
The difference between the disagreement of $\mathbf{s}$ and $\mathbf{s}_0$ after running the dynamics is easily seen to be
\begin{equation}
    2(\Delta \mathbf{s})^T(I+L)^{-1}L(I+L)^{-1}\mathbf{s}_0+(\Delta \mathbf{s})^T(I+L)^{-1}L(I+L)^{-1}(\Delta \mathbf{s}),
\end{equation}
where $\Delta \mathbf{s}=\mathbf{s}-\mathbf{s}_0$. Applying Cauchy-Schwarz to the last term with the result in Theorem \ref{thm:dis} gives an upper bound of
\begin{equation}
   \|(I+L)^{-1}L(I+L)^{-1}\|\|\Delta \mathbf{s}\|_2^2\leq \lambda_{max}(\Sigma)k.
\end{equation}
The first term can be bounded in a couple of ways: using H\"older's inequality and Lemma \ref{lem:racz1} and \ref{lem:racz2}, one obtains
\begin{equation}
    2(\Delta \mathbf{s})^T(I+L)^{-1}L(I+L)^{-1}\mathbf{s}_0\leq 2\|(I+L)^{-1} (\Delta \mathbf{s})\|_1 \|L(I+L)^{-1}\mathbf{s}_0\|_{\infty}\leq 2kd_{\max}.
\end{equation}

Another way to bound this first term is simply using Cauchy-Schwarz and Theorem \ref{thm:dis} to get
\begin{equation}
    2(\Delta \mathbf{s})^T(I+L)^{-1}L(I+L)^{-1}\mathbf{s}_0\leq 2\|(\Delta \mathbf{s})\|_2 \|(I+L)^{-1}L(I+L)^{-1}\mathbf{s}_0\|_{2}\leq 2\cdot \lambda_{max}(\Sigma)\sqrt{k}\|\mathbf{s}_0\|_2\leq 2\cdot \lambda_{max}(\Sigma)\sqrt{nk}.
\end{equation}
Therefore, we conclude that the increase in  disagreement is bounded by
\begin{equation}
    \lambda_{max}(\Sigma)k+\sqrt{k}\min\{2 d_{max}\sqrt{k},2\cdot \lambda_{max}(\Sigma)\sqrt{n}\}.
\end{equation}
\end{proof}

In \cite{chen2020network}, the authors achieve a bound of $8d_{max}k$. From this combined result here, coupled with the previous sections, one obtains an improvement of the constants and dependence on $d_{max}$; moreover, this shows that the dependence on $d_{max}$ is unnecessary for $d_{max}=\Omega(\sqrt{n/k})$. In fact, though the bound attained here seems to be worse in the setting where $d_{max}\to 0$, this is not so; it is trivial to see that $\frac{x}{(1+x)^2}\leq x$ for $x\geq 0$, so that $\lambda_{max}(\Sigma)\leq \lambda_{max}(L)$. By the Gershgorin circle theorem, one also has $\lambda_{max}(L)\leq 2d_{max}$, thus giving an improvement in the small $d_{max}$ setting as well.

We note that the linear dependence on $k$ is indeed necessary in general.\footnote{By this, we just mean that there \emph{exists} instances of $L$ and $\mathbf{s}_0$ where the dependence in $k$ is provably linear. It is easy to see that for any $L$, there exists some $\mathbf{s}_0$ where the dependence on $k$ is trivial. For any $L$, let $\mathbf{s}_0$ be a vector that maximizes disagreement in the cube $[0,1]^n$ (as the maximum of a convex function, it will lie in $\{0,1\}^n$, but that is not used in the argument). Then no perturbation inside the cube can improve it, hence the dependence on $k$ is trivial. It would be interesting, as noted in \cite{chen2020network}, if one could show that the dependence on $k$ is linear for ``typical'' instances, as suggested by their empirical results.} For instance, consider the family of graphs $L_n := \frac{1}{n} L_{K_n}$. All nontrivial eigenvalues are located at $\lambda = 1$, hence are a worst-case example for disagreement. Consider the vectors $\mathbf{s}^k$ obtained by setting the first $k$ entries to $1$, with the rest $0$. It is easy to see that $\mathbf{s}^k$ can be decomposed as
\begin{equation}
    \mathbf{s}^k = \frac{k}{\sqrt{n}}\frac{\mathbf{1}}{\sqrt{n}}+\bigg(\sqrt{\frac{k(n-k)}{n}}\bigg)\mathbf{r}^k,
\end{equation}
where $\mathbf{r}^k$ is a unit vector orthogonal to $\mathbf{1}$, hence an eigenvector of $\Sigma$ with eigenvalue $1/4$. Thus, if we instantiate this sparsity-constrained problem with $\mathbf{s}_0=\mathbf{0}$, we have
\begin{equation}
    \mathbf{s}^k\Sigma\mathbf{s}^k = \frac{k(n-k)}{4n}=(1-k/n)\cdot\frac{k}{4}.
\end{equation}
In particular, for $k<n/2$ for instance, this is $\Omega(k)$.

\subsubsection{$\ell_{\infty}$-Constrained Adversary}

Now consider the problem of the adversary choosing weights subject to an $\ell_{\infty}$ constraint; for convenience, we will write the disagreement problem\footnote{The corresponding results for any other positive semidefinite objective, like the others we considered before, follow in exactly the same way.} as
\begin{equation}
    \max_{\mathbf{s}\in \mathbb{R}^n:\|\mathbf{s}\|_{\infty} = 1} \mathbf{s}^T(I+L)^{-1}L(I+L)^{-1}\mathbf{s}.
\end{equation}
For simplicity, again write the matrix as $\Sigma$. The first simple observation is that the optimal adversary strategy $\mathbf{s}^*$ will satisfy $\vert \mathbf{s}^*_i\vert=1$ for all $i$; indeed, if one fixes $\mathbf{s}^*_{-i}$, then the objective is easily seen to be linear in $\mathbf{s}_i$, and therefore attains a maximum for $\mathbf{s}^*_i\in \{-1,1\}$. In particular, the problem can be rewritten as
\begin{equation}
\label{eq:advinf}
    \max_{\mathbf{s}\in \mathbb{R}^n: \forall i\,\mathbf{s}_i \in  \{\pm1\}} \mathbf{s}^T\Sigma \mathbf{s}.
\end{equation}
In general, there does not appear to be any analytic solution to this problem.\footnote{As we show in Remark \ref{rmk:nphard} in Section \ref{section:mixedgraph}, the generalization of this objective to the mixed-graph setting is actually easily seen to be $\mathsf{NP}$-hard.} One always has the trivial spectral bound $\lambda_n(\Sigma)\cdot n$, just using the fact that the $\ell_2$-norm of the candidate solutions are all $\sqrt{n}$. However, by leveraging deep results from functional analysis that have long been fruitfully applied to combinatorial optimization, it is possible to give an efficient constant factor approximation:
\begin{thm}
There exists a polynomial-time algorithm that computes a $\pi/2$-approximation to (\ref{eq:advinf}); moreover, a vector $s' \in \{-1,1\}^n$ attaining this bound can be obtained efficiently.\footnote{We ignore issues of bit complexity; if desired, one may assume that $L$ is given as a rational matrix and that we will be content with solutions up to $\epsilon$ accuracy.}
\end{thm}
\begin{proof}
To prove the theorem, we apply \emph{Grothendieck's inequality} \cite{alon2004approximating}, which asserts for any positive semidefinite matrix $Z$,
\begin{equation}
    \max_{\mathbf{x}\in \mathbb{R}^n: \forall i\,x_i \in  \{\pm1\}} \mathbf{x}^TZ \mathbf{x}\leq \max_{\mathbf{x}_1,\ldots,\mathbf{x}_n\in S^{n-1}} \sum_{i,j=1}^n Z_{ij}\langle \mathbf{x}_i,\mathbf{x}_j\rangle\leq \frac{\pi}{2} \max_{\mathbf{x}\in \mathbb{R}^n: \forall i\,x_i \in  \{\pm1\}} \mathbf{x}^TZ \mathbf{x},
\end{equation}
where $S^{n-1}=\{\mathbf{x}\in \mathbb{R}^n: \|\mathbf{x}\|_2=1\}$. It is well-known that the middle expression is easily attained as the solution to the semidefinite program
\begin{align*}
\max_{X\in \mathbb{R}^{n\times n}} &\text{Tr}(ZX)\\
    \text{subject to } X_{ii} &= 1\\
    X&\succeq 0,
\end{align*}
and therefore can be solved to arbitrary accuracy in polynomial time. By setting $Z$ to be our matrix $\Sigma$ and applying this result, it follows that the solution to this problem thus gives a $\frac{\pi}{2}$-approximation to the adversary's problem. We remark that, following the approach for the $\mathsf{MAXCUT}$ problem in \cite{delorme1993laplacian}, one can appeal to strong SDP duality to obtain an equivalent, and slightly more analytic, spectral upper bound of $n\cdot \min_{\mathbf{x}:\mathbf{x}^T\mathbf{1}=0} \lambda_{max}(\Sigma-\text{diag}(\mathbf{x}))$. This clearly improves on the trivial spectral bound given before, though it is not clear in general what good test vectors $\mathbf{x}$ are to give an explicit better bound.

To actually obtain a $\pm 1$ vector attaining this bound compared to the optimum of the SDP, Alon and Naor \cite{alon2004approximating} show that randomized hyperplane rounding applied to the vectors attaining the SDP optimum can be used to obtain a vector $\mathbf{s}$ that is at least $\frac{2}{\pi}$ of the SDP optimum in expectation, hence of the original problem .
\end{proof}

\subsubsection{$\ell_1$-Constrained Adversary}
One may also consider an adversary that is bounded in $\ell_1$, which may be viewed as a more natural restriction. However, in this case, the problem becomes trivial; again, we consider the disagreement problem with matrix $\Sigma$, but all these results carry over for any other positive semidefinite objective.
\begin{thm}
Suppose the adversary is now $\ell_1$-constrained, so that the optimization problem is
\begin{equation}
    \max_{\mathbf{s}\in \mathbb{R}^n: \|\mathbf{s}\|_1=1} \mathbf{s}^T\Sigma \mathbf{s}.
\end{equation}
Then, an optimal adversary strategy is simply to set $s=\mathbf{e}_i$ where $i\in \arg\max_{j\in [n]} \Sigma_{jj}$ (that is, to put all their budget on an index with largest diagonal term in $\Sigma$).

In this case, denoting the adversary's power as $D$, we have
\begin{equation}
    \frac{1}{n}\sum_{i=1}^n \frac{\lambda_i(L)}{(1+\lambda_i(L))^2}\leq D\leq \max_{i\in [n]} \frac{\lambda_i(L)}{(1+\lambda_i(L))^2}\leq \frac{1}{4},
\end{equation}
where the lower bound is sharp in the case of vertex-transitive graphs.
\end{thm}

\begin{proof}
For an optimal adversary strategy, simply note that the $\ell_1$-ball is the convex hull of the set $\{\pm \mathbf{e}_1,\ldots,\pm \mathbf{e}_n\}$. As the maximum of a convex function over a convex set is attained at an extreme point, it suffices to consider this set of strategies; by homogeneity, it suffices to just consider $\{\mathbf{e}_1,\ldots,\mathbf{e}_n\}$. But when substituting these terms in the optimization problem, one then recovers the diagonal elements of $\Sigma$, so the adversary may simply choose the largest such term.

The upper bound follows from the $\ell_2$ bound we gave before, using the fact that the $\ell_1$ unit ball is contained in the $\ell_2$ unit ball. For the lower bound, observe that
\begin{equation}
    \sum_{i=1}^n \mathbf{e}_i^T\Sigma \mathbf{e}_i=\sum_{i=1}^n \Sigma_{ii}=\text{Tr}(\Sigma).
\end{equation}
As the trace of any matrix is equal to the sum of the eigenvalues, it follows there must exist a diagonal element that is at least the average of the eigenvalues. 

To see the tightness of the lower bound, consider any vertex-transitive graph. By definition, the adjacency matrix is unchanged under the action of a transitive subgroup of the symmetric group. Any vertex-transitive graph is regular, so the degree matrix is a multiple of the identity, hence also unchanged under the action of any permutation; together, these imply that the Laplacian of the graph must be invariant under the action of a transitive subgroup of the symmetric group. As a result, any rational expression of the Laplacian is also invariant. By applying suitable automorphisms of the graph, this implies that any two diagonal elements of any rational expression of the Laplacian must be equal, in which case they are all equal to the average of the eigenvalues by the trace identity.
\end{proof}

\section{Defending the Network}
\label{section:defense}
In this section, we consider the problem of \emph{defending} a given network from adversarial perturbations like those considered above. We will view this as a two-player min-max game; first, a network defender will choose how to set some qualitative feature of the network subject to normalization constraints modeling the resource limitation of the defender. Then, the adversary performs the above maximization problem with this choice of settings. The goal of the defender is to choose 
a setting to minimize the cost of the resulting system (e.g., the measure of disagreement), knowing that the adversary will optimize for this choice. In this section, we show that, in one such formulation, the defender can efficiently do this via solving an appropriate convex optimization problem.

We have generally adopted the convention that the network topology is basically fixed; it is unrealistic to substantively change a real-world network structure. Therefore, in this formulation, the network defender chooses how to vary the cost of the adversary in changing initial opinions of different nodes. That is, the network defender can choose to weigh each node differently for example, by lessening their exposure to misinformation, so that the adversary pays different costs for perturbing different nodes.

Formally, we consider the following problem: suppose the network defender is resource limited according to a function $h:\mathbb{R}^n\to \mathbb{R}$, such as the $\ell_1$-norm, and is permitted to change node-weights with the restriction of $h(\mathbf{w})=h(\mathbf{1})$. We will consider, under this resource constraint, what the defender's optimal choice of $\mathbf{w}$ is. More generally, we will assume that $h$ is nonnegative, convex, and radially increasing and homogeneous (i.e. for $\alpha\geq 0$, $h(\alpha\mathbf{x})=g(\alpha)h(\mathbf{x})$, with $g:\mathbb{R}_+\to \mathbb{R}_+$ an increasing function), as well as a function $f:\mathbb{R}\to \mathbb{R}$ as in (\ref{eq:advproblem}) that induces a positive semidefinite quadratic form. Then the network defender must solve the following optimization problem:
\begin{equation}
\label{eq:defenderproblem}
    \min_{\mathbf{w}\in \mathbb{R}^n_{>0}: h(\mathbf{w})=h(\mathbf{1})}\bigg\{\max_{\mathbf{s}\in \mathbb{R}^n: \|\mathbf{s}\|_{\mathbf{w}}\leq R} \mathbf{s}^T (I+L)^{-1} f(L)(I+L)^{-1} \mathbf{s}\bigg\}.
\end{equation}
In words, the network defender chooses a \emph{weighted} $\ell_2$-norm on the nodes of the network under the resource constraint modeled by $h$ that specifies costs of influencing each individual in the network heterogeneously; with these weights and the same fixed budget $R$ as before, the adversary then optimizes their objective. For instance, if $h$ is the $\ell_1$-norm, then this normalization imposes that $\sum_{i=1}^n w_i=\sum_{i=1}^n 1=n$, so that the sum of weights on the nodes for the adversary is the same as for the regular $\ell_2$ norm. Other natural choices for $h$ include any norm on $\mathbb{R}^n$ or any sum of squares of linear expressions. As a result, $\mathbf{w}=\mathbf{1}$ is a valid choice of the network defender, in which case the inner maximization corresponds to the largest eigenvalue of the relevant quadratic form as we have seen above. However, for other choices of $\mathbf{w}$, the inner maximization does not have the same interpretation and moreover, will not usually admit a clean analytical expression as the maximization of a convex objective.

Here, we show that despite this difficulty, this 
can be reduced to convex optimization via a geometric argument. For convenience, set $R=1$; this is without loss of generality as the inner maximization is homogeneous. Consider the following procedure:
\begin{enumerate}
    \item Solve the following convex program with positive semidefinite constraints:
    \begin{gather*}
        \min_W h(\text{diag}(W))\\
        \text{subject to } 0\preceq (I+L)^{-1} f(L) (I+L)^{-1}\preceq W\\
        W_{ij} = 0, \quad \forall i\neq j.
    \end{gather*}
    
    \item Set $\mathbf{w}'=\text{diag}(W^*)$, where $W^*$ is a solution to the above convex program.
    
    \item Let $t\geq 0$ be such that $h(t\mathbf{w}')=h(\mathbf{1})$.
    
    \item Set $\mathbf{w}^*=t\mathbf{w}'$.
\end{enumerate}

We now show that this procedure gives the optimal setting of $\mathbf{w}$.
\begin{thm}
Under the restrictions on $h$ and $f$, the above algorithm yields the optimal value of the problem given by (\ref{eq:defenderproblem}).
\end{thm}
\begin{proof}
Let $h$ and $f$ be as required and write $\Sigma=(I+L)^{-1}f(L)(I+L)^{-1}$; by our assumption on $f$, $\Sigma\succeq 0$. For any fixed choice of $\mathbf{w}$, the adversary's optimal choice of $\mathbf{s}\in \mathbb{R}^n$ is obtained by finding the largest level set of the function $\mathbf{x}^T \Sigma\mathbf{x}$ that nontrivially intersects the ellipsoid $\|\mathbf{x}\|_{\mathbf{w}}\leq 1$. Equivalently, the optimal value of the adversary is the smallest level set of $\mathbf{x}^T \Sigma\mathbf{x}$ that contains the unit ball of the norm induced by $\|\cdot\|_{\mathbf{w}}$. In particular, the optimal value of the inner maximization for fixed $\mathbf{w}$ is the smallest value $K\geq 0$ such that
\begin{equation}
\label{eq:optimality}
    \{\mathbf{x}\in \mathbb{R}^n: \|\mathbf{x}\|_{\mathbf{w}}\leq 1\}\subseteq \{\mathbf{x}\in \mathbb{R}^n: \mathbf{x}^T\Sigma \mathbf{x}\leq K\}
\end{equation}

 By the restriction on $W$ to being diagonal in the above convex program (with necessarily nonnegative diagonal entries by the PSD constraint), recall that $W\succeq \Sigma$ if and only if for all $\mathbf{x}\in \mathbb{R}^n$, 
\begin{equation*}
    \|\mathbf{x}\|_{\mathbf{w}'}^2=\mathbf{x}^TW\mathbf{x}\geq \mathbf{x}^T\Sigma \mathbf{x};
\end{equation*}
where $\mathbf{w}'=\text{diag}(W)$; geometrically, this is equivalent to the containment 
\begin{equation}
\label{eq:feasibility}
    \{\mathbf{x}\in \mathbb{R}^n: \|\mathbf{x}\|_{\mathbf{w}'}\leq 1\}\subseteq \{\mathbf{x}\in \mathbb{R}^n: \mathbf{x}^T \Sigma \mathbf{x}\leq 1\}.
\end{equation}
 In particular, this means that $\Sigma\preceq W$ if and only if the unit ball of $\|\cdot \|_{\mathbf{w}'}$ is contained in the unit ball of the (semi)-norm induced by $\Sigma$. Let $W^*$ and $\mathbf{w}'=\text{diag}(W^*)$ be as stated, and let $t\geq 0$ be such that $h(t\mathbf{w}')=h(\mathbf{1})$. By the minimality of $\mathbf{w}'$ as well as the positive homogeneity, this implies that if the optimal value of the inner maximization for (\ref{eq:defenderproblem}) using $t\mathbf{w}'$ is $K$, then using homogeneity of the containments:
 \begin{gather*}
     \{\mathbf{x}\in \mathbb{R}^n: \|\mathbf{x}\|_{t\mathbf{w}'}\leq \frac{1}{\sqrt{K}}\}\subseteq \{\mathbf{x}\in \mathbb{R}^n: \mathbf{x}^T \Sigma \mathbf{x}\leq 1\}\\
     \iff \{\mathbf{x}\in \mathbb{R}^n: \|\mathbf{x}\|_{\mathbf{w}'}\leq \frac{1}{\sqrt{t\cdot K}}\}\subseteq \{\mathbf{x}\in \mathbb{R}^n: \mathbf{x}^T \Sigma \mathbf{x}\leq 1\};
 \end{gather*}
 as such, $K=1/t$.

Suppose now for a contradiction that an optimizer $\mathbf{w}^*$ of (\ref{eq:advproblem}) has strictly smaller objective value $K^*<1/t$ than that of $t\mathbf{w}'$. By (\ref{eq:optimality}), this is equivalent to
\begin{gather*}
    \{\mathbf{x}\in \mathbb{R}^n: \|\mathbf{x}\|_{\mathbf{w}^*}\leq \frac{1}{\sqrt{K^*}}\}\subseteq \{\mathbf{x}\in \mathbb{R}^n: \mathbf{x}^T \Sigma \mathbf{x}\leq 1\}\\
    \iff \{\mathbf{x}\in \mathbb{R}^n: \|\mathbf{x}\|_{K^*\mathbf{w}^*}\leq 1\}\subseteq \{\mathbf{x}\in \mathbb{R}^n: \mathbf{x}^T \Sigma \mathbf{x}\leq 1\}.
\end{gather*}
Evidently, $K^*\mathbf{w}^*$ satisfies (\ref{eq:feasibility}) yet
\begin{equation*}
    h(K^*\mathbf{w}^*)=g(K^*)h(\mathbf{w}^*)<g(1/t)h(t\mathbf{w}')=h(\mathbf{w}').
\end{equation*}
This violates the optimality of $\mathbf{w}'$, yielding the desired contradiction.
\end{proof}

Note that for certain choices of $h(\cdot)$, the above can be written as a standard semidefinite program. For instance, if $h(\mathbf{x})=\|\mathbf{x}\|_1=\sum_{i=1}^n\vert x_i\vert$, then the above is indeed a regular semidefinite program. If $h$ is instead the squared $\ell_2$-norm, one can similarly write it as a semidefinite program by adding extra positive semidefinite constraints and exploiting Schur complements; we omit the details here.

\section{Mixed-Graph Objectives}
\label{section:mixedgraph}

 In the previous sections, we have connected an adversary's ability to induce discord in a network with the spectral theory of the underlying graph, as well as considered ways to defend against these attacks. In each of these settings, these results suggest that the opinion dynamics of the network will necessarily ``soften'' the effect of these attacks, as  the disagreement is measured on the edges of the same network that dictate the opinion dynamics. However, one potential explanation of the success of recent adversarial attacks described in the introduction is that the opinion formation graph and disagreement measurement graph need not be the same, and may not even look similar. For instance, opinion formation may take place on the ``online'' network, via social media, while the disagreement the adversary cares about maximizing may be measured with respect to ``real-world'' connections. When this occurs, it need not be the case that the opinion dynamics implicitly equilibrate disagreements measured along the latter graph. When opinion formation and the disagreement graph look quite different, one expects that an adversary will be able to induce significantly more disagreement.

In this section, we explore the degree to which the adversary's power can increase when the opinion formation graph and measurement graph become independent. First, we provide nontrivial examples that show that, in some cases, having a different graph for the opinion dynamics and for the disagreement measurement can actually \emph{reduce the adversary's power to induce disagreement}; however, we provide a general lower bound that indicates that typically, the adversary will not be much worse off, if at all. We then show that the relevant relationship that will determine when an adversary gains extra power is an appropriate notion of \emph{spectral similarity} between the two graphs, not necessarily physical similarity. We conclude this section by showing concretely how a large cut misalignment in the graphs will enable an adversary to induce disagreement far beyond what is possible in the spectral theory in the single-graph setting.

Formally, we generalize the previous sections as follows: suppose that there are now \emph{two} relevant graph structures on $[n]$, $G_1$ and $G_2$, with associated Laplacians $L$ and $M$ (which will be the ``measurement" Laplacian). The first graph, $G_1$, is the graph structure on which the Friedkin-Johnsen dynamics take place while the second graph, $G_2$, is the graph where disagreement is measured. In this setting, the adversary chooses initial opinions to maximize the following objective (setting $R=1$ for notational ease):
\begin{equation*}
    \max_{\mathbf{s}\in \mathbb{R}^n:\|\mathbf{s}\|_2\leq 1} \mathbf{s}^T(I+L)^{-1}M(I+L)^{-1}\mathbf{s}=\lambda_{\text{max}}((I+L)^{-1}M(I+L)^{-1}).
\end{equation*}
Typically, we will be interested in settings where $L$ and $M$ are of comparable size, meaning similar total edge weight. If not, say if $L$ has much larger edge weight than $M$, then the effects of the opinion dynamics will cause all opinions to smooth to a much larger degree compared to the measurement measured via $M$, so the problem becomes degenerate though not for a theoretically interesting reason.

\jgedit{\begin{remark}
\label{rmk:nphard}
We note that the $\ell_{\infty}$ version of the adversary problem in this setting is easily seen to be $\mathsf{NP}$-hard. The reason is that one can let $M$ be the Laplacian of any graph, and simply set $L$ to be the trivial graph with no edges, at which point the problem can be shown to be precisely $\mathsf{MAXCUT}$ (see Section \ref{section:sim} for the relation to cuts). As such, by appealing to various complexity-theoretic assumptions, one can easily establish hardness-of-approximation results for this version of the problem \cite{khot2007optimal, haastad2001some}. In particular, it is $\mathsf{NP}$-hard to attain an approximation within better than a $17/16$-factor of the optimum in general and $\mathsf{UGC}$-hard to obtain a solution within a factor of better than $\alpha_{GW}\approx 1.14$ of the optimum in general with an $\ell_{\infty}$-adversary.
\end{remark}}

In general, it is not obvious how to connect the spectral structure of the above matrix with the spectral properties of the two underlying graphs, unless in the special case where the two graphs commute (and therefore, share an eigenbasis); this is indeed possible in certain special cases.

\begin{example}
Suppose $G_1$ is a $d$-regular, unweighted graph, and let $G_2$ be the $n-d-1$-regular, unweighted complementary graph. Then one can check that $M=nI-J-L$, where $J$ is the matrix of all-ones. Every matrix on the right side shares an eigenbasis, hence $L$ and $M$ commute. Similarly, suppose that $G_1$ or $G_2$ is a scaled version of the unweighted complete graph $K_n$. Then it is easy to check that $L_{K_n}$ commutes with every graph Laplacian $L$ as they will share a common eigenbasis.
\end{example}

One might suspect that, fixing $M$, $L=M$ is the optimal choice of graph Laplacian (subject to normalization) to minimize the amount of disagreement an adversary can induce. That is, letting $\mathcal{L}$ be the set of graph Laplacians subject to the edge normalization $\text{Tr}(L)=\text{Tr}(M)$, one might guess that
\begin{equation*}
    M\in \arg\min_{L\in \mathcal{L}} \lambda_{\text{max}}((I+L)^{-1}M(I+L)^{-1}).
\end{equation*}
However, this does not hold in general, via the following simple construction.

\begin{example}
Suppose that $G_2$ is a complete graph, in the sense that for some $\epsilon>0$, all off-diagonals of the Laplacian satisfy
\begin{equation*}
    M(i,j)< -\epsilon.
\end{equation*}
If we write out $M$ in the eigenbasis as
\begin{equation*}
    M=\sum_{i=2}^n \lambda_i(M)\mathbf{v}_i\mathbf{v}_i^T
\end{equation*}
and further suppose for simplicity that $1<\lambda_2(M)<\lambda_3(M)$, so that all nonzero eigenvalues lie on the right side of the peak of the function $f(x)=x/(1+x)^2$ at $x=1$. If we then set $L=M$, we would get an objective value of
\begin{equation*}
    \frac{\lambda_2(M)}{(1+\lambda_2(M))^2},
\end{equation*}
as we have seen before.
But consider instead the matrix
\begin{equation*}
    L = (\lambda_2(M)+\eta)\mathbf{v}_2\mathbf{v}_2^T+(\lambda_3(M)-\eta)\mathbf{v}_3\mathbf{v}_3^T+\sum_{i=4}^n \lambda_i(M)\mathbf{v}_i\mathbf{v}_i^T
\end{equation*}
for some $\eta>0$ sufficiently small (depending on $\epsilon$). It is easy to see that $L\in \mathcal{L}$ as the sum of eigenvalues, and therefore the trace, is constant, and moreover, $L$ will still be a Laplacian (with nonpositive off-diagonal entries) of some other graph by continuity. As $L$ and $M$ share an eigenbasis, it is easy to see that for $\eta$ small enough, 
\begin{equation*}
    \lambda_{\text{max}}((I+L)^{-1}M(I+L)^{-1})=\frac{\lambda_2(M)}{(1+\lambda_2(M)+\eta)^2}
\end{equation*}
which is strictly smaller than if $L=M$.
\end{example}

\begin{example}
In a more interesting example, suppose now that we further require that $(L)_{i,i}=(M)_{i,i}$ for each $i$. This means that each node has the same weighted degree in both graphs, which in particular implies they have the same trace. First, consider $G_2=C_4$, the four node unweighted cycle graph. One can numerically check that among all graphs $G_1$ satisfying this normalization, the mixed objective $(I+L)^{-1}M(I+L)^{-1}$ is minimized when $G_1$ is a weighted complete graph where each edge in the cycle has weight reduced from $1$ to approximately $.89$, and the remaining two edges are increased from $0$ to $.22$. The mixed objective has value approximately $0.1929$, whereas the single-graph objective $(I+M)^{-1}M(I+M)^{-1}$ has largest eigenvalue $2/9\approx .22$. For comparison, when $G_1$ is instead set to the appropriately scaled copy of the complete graph, the mixed-objective actually \emph{rises} to $.2975$. On the other hand, when $G_2=P_4$, the unweighted path graph on four nodes, it is numerically optimal for itself under the mixed-graph objective for all graphs satisfying the degree constraint. We are unaware of an analytic reason why this holds.
\end{example}

However, one expects that these examples are largely pathological. As a first approximation to controlling this quantity, our first result is the following general bound for positive semidefinite matrices. The idea is to apply the Courant-Fischer theorem to subspaces spanned by the eigenvectors of the two matrices to lower bound the spectral norm of the product of matrices.
\begin{lem}
\label{lem:matbound}
Let $B,C\in \mathbb{R}^{n\times n}$ be positive semidefinite matrices with eigenvalues in increasing order. Then
\begin{equation*}
    \max_{k\leq n}\bigg\{\lambda_{n-k+1}(C)^2\lambda_{k}(B)\bigg\}\leq \lambda_n(CBC)= \|CBC\|_2\leq \lambda_n(C)^2\lambda_n(B).
\end{equation*}
\end{lem}

\begin{proof}
Note that $CBC\succeq 0$ by the fact $B\succeq 0$. The upper bound follows directly from the submultiplicativity of the operator norm, which for symmetric positive-semidefinite matrices is just the top eigenvalues.

For the lower bound, we use the Courant-Fischer theorem. First, note that if $\lambda_{n-k+1}(C)=0$, the result is trivial, so suppose it is strictly positive. Let $U$ be the linear subspace spanned by the top $k$ eigenvectors of $C$, and let $V$ be the subspace spanned by the top $n-k+1$ eigenvectors of $B$. These subspaces must intersect non-trivially by a simple dimension argument, so there exists some $\mathbf{z}\in U\cap V$ with unit length. Note that $U$ is an invariant subspace for $C$, and moreover, $C$ is bijective on $U$ by the nondegeneracy of $\lambda_{n-k+1}(C)$. Now, let $\mathbf{x}=C^{-1}\mathbf{z}$, where we view $C^{-1}$ as restricted to $U$. By the variational formula of $\lambda_n$,
\begin{equation*}
    \lambda_n(CBC)\geq \frac{\mathbf{x}^TCBC\mathbf{x}}{\|\mathbf{x}\|_2^2}=\frac{\mathbf{z}^TB\mathbf{z}}{\mathbf{z}^TC^{-2}\mathbf{z}}\geq \lambda_{k}(B)\frac{\|\mathbf{z}\|_2^2}{\|C^{-1}\mathbf{z}\|_2^2}=\frac{\lambda_{k}(B)}{\|C^{-1}\mathbf{z}\|_2^2}
\end{equation*}
The second inequality follows from Courant-Fischer, as $\mathbf{z}$ lies in the span of the top $n-k+1$ eigenvectors of $B$ by assumption, so the quadratic form in the numerator gives at least $\lambda_{n-(n-k+1)+1}(B)\|\mathbf{z}\|^2=\lambda_{k}(B)\|\mathbf{z}\|^2$.
Then $\|C^{-1}\mathbf{z}\|^2\leq \lambda_{\text{max}}(C^{-1})^2\|\mathbf{z}\|^2=\frac{\|\mathbf{z}\|^2}{\lambda_{n-k+1}(C)^2}$,
as the largest eigenvalue of $C^{-1}$ restricted to $U$ is the inverse of the smallest eigenvalue of $C$ restricted to $U$. Plugging this in gives the desired inequality. As this holds for all $k\leq n$, it holds for the maximum.
\end{proof}

From this simple lemma, one can immediately obtain a lower bound in the adversary's optimization problem in the mixed-graph setting.
\begin{corollary}
Let $L,M$ be as above. Then 
\begin{equation*}
    \lambda_{\text{max}}((I+L)^{-1}M(I+L)^{-1})\geq \max_{1\leq k\leq n}\frac{\lambda_{k}(M)}{(1+\lambda_{k}(L))^2}.
\end{equation*}
\end{corollary}
\begin{proof}
This is immediate from the previous lemma, with $B=M$ and $C=(I+L)^{-1}$, simply noting that
\begin{equation*}
    \lambda_{n-k+1}((I+L)^{-1})=\frac{1}{1+\lambda_{k}(L)}.
\end{equation*}
\end{proof}

Using just this lower bound, in the special case where $M=L_{K_n}$, the Laplacian of the unweighted complete graph on $n$ nodes, it follows that the optimal choice of $L$ subject to having the same trace as $L_{K_n}$ to minimize the mixed-graph objective is just $L_{K_n}$ \emph{itself}. This holds because for any graph Laplacian $L$ satisfying $\text{Tr}(L)=\text{Tr}(L_{K_n})$, a similar argument to that of Corollary \ref{cor:completeoptimal} implies that some nontrivial eigenvalue of $L_{1}$ must be at most $n$. The lower bound of the previous corollary then asserts that the mixed-graph objective can only increase, with equality if and only if $L=L_{K_n}$.

Moreover, the previous corollary asserts that if the eigenvalues of $L$ and $M$ are only \emph{numerically} similar in the appropriate ordering, then necessarily the objective value will be \emph{approximately at least} the corresponding objective value we considered in Section \ref{section:disagreement}. \jgedit{Explicitly, this will arise for any graphs with \emph{cospectral} Laplacians; for instance, any isomorphic graphs will have this property (so in particular, if $M$ and $L$ differ by a permutation), as will any strongly regular graphs with same parameters.} This suggests that while we have shown explicit examples where having two distinct matrices can even reduce the adversary's power, this case ought be viewed as rather pathological.

\subsection{Spectral Similarity}
\label{section:sim}
The above analysis relied only on a general lower bound involving positive semidefinite matrices. Next, we aim to characterize the relevant structure of $L$ and $M$ that causes the objective function to remain quite close to the value in the single-graph case, and similarly when the objective function will increase. The former case will indicate that an adversary gains little benefit from the misalignment of $G_1$ and $G_2$, while the latter case corresponds to an underlying network that can be sharply exploited to induce large disagreement. Intuitively, if $L\approx M$ component-wise, then
\begin{equation*}
    \lambda_{\text{max}}((I+L)^{-1}M(I+L)^{-1})\approx \lambda_{\text{max}}((I+M)^{-1}M(I+M)^{-1}),
\end{equation*}
by the continuity of matrix inverses and eigenvalues. Before proceeding, we need a definition:

\begin{defn}
For any graph $G=(V,E,w)$ and $S,T\subseteq V$ such that $S\cap T=\emptyset$, we define 
\begin{equation*}
    \textsf{cut}_G(S,T)=\sum_{i\in S, j\in T} w_G(i,j).
\end{equation*}
We will write $\textsf{cut}_G(S):=\textsf{cut}_G(S,S^c)$.

For any subset $S\subseteq [n]$, we write $\chi_S$ for the $\pm 1$ indicator vector of $S$, i.e. $\chi_S(i)=1$ if $i\in S$ and $-1$ if $i\not\in S$. Then it is easy to see that $\|\chi_S\|^2=n$, and that for any graph $G$ with Laplacian $L$,
\begin{align*}
    \chi_S^T L\chi_s&=\sum_{(i,j)\in E_i} w_{G}(i,j)(\chi_S(i)-\chi_S(j))^2\\
    &=4\sum_{i\in S, j\not\in S} w_{G}(i,j)=4\mathsf{cut}_{G}(S).
\end{align*}
\end{defn}

We can now provide a quantitative form of this assertion:

\begin{thm}
Let $G_1,G_2$ be $n$-node graphs with Laplacians $M$ and $L$, respectively. Suppose that the following holds for some parameters $\eta,\gamma,\epsilon>0$:
\begin{enumerate}
    \item For each $i\in [n]$, the weighted symmetric difference of their neighborhoods is bounded by $\eta$, i.e. for all $i\in [n]$
    \begin{equation}
        \sum_{j\neq i} \vert w_1(i,j)-w_2(i,j)\vert\leq \eta.
    \end{equation}
    
    \item For all $i\in [n]$, the absolute difference in weighted degrees of $i$ in $G_1$ and $G_2$ is at most $\gamma$.
    
    \item For any disjoint subsets $S,T\subseteq [n]$, we have
    \begin{equation}
        \vert \textsf{cut}_{G_1}(S,T)-\textsf{cut}_{G_2}(S,T)\vert \leq \epsilon \sqrt{\vert S\vert\vert T\vert}.
    \end{equation}
\end{enumerate}
Then, we have
\begin{equation}
    \max_{i\in [n]} \frac{\lambda_i(M)-2\Delta}{(1+\lambda_i(M)+\Delta)^2}\leq \lambda_{n}((I+L)^{-1}M(I+L)^{-1}) \leq \max_{i\in [n]} \frac{\lambda_i(M)+2\Delta}{(1+\lambda_i(M)-\Delta)^2},
\end{equation}
where $\Delta=O(\epsilon \ln(\eta/\epsilon)+\gamma)$. In particular, if all nodes have the same weighted degree in both $G_1$ and $G_2$, $\Delta = O(\epsilon \ln(\eta/\epsilon))$.
\end{thm}

\begin{remark}
Before proceeding with the proof, note that if one only assumes the first condition above, and even if every node has the same degree in both graphs (therefore satisfying the second condition with $\gamma=0$), the best bound one can generically get on the spectral radius of $M-L$ is $O(\eta)$ using the fact that the largest eigenvalue of a matrix is at most the largest $\ell_1$ norm of a row. When the combinatorial structures are assumed to be very similar along every subset, as is done here, the dependence on $\eta$ becomes \emph{logarithmic}, and gains from the closeness in the $\epsilon$ term as well.
\end{remark}

\begin{proof}
First, by our assumptions, we may apply Lemma 3.3 of Bilu and Linial \cite{bilu2006lifts}, where we just note that if $u$ is the $\{0,1\}$ indicator of $S$ and $v$ is the $\{0,1\}$ indicator for $T$ for some disjoint subsets $S,T\subseteq [n]$, then 
\begin{equation}
    \vert u^T(M-L)v\vert= \vert \textsf{cut}_{G_1}(S,T)-\textsf{cut}_{G_2}(S,T)\vert.
\end{equation}
We also note that by inspecting the proof of that lemma, one can apply our condition (2) with parameter $\gamma$ instead of $O(\epsilon \ln(\eta/\epsilon))$ by just paying it in the bound, from which it follows that $\Delta:=\|M-L\|\leq O(\epsilon \ln (\eta/\epsilon)+\gamma)$.

From this, we have
\begin{equation}
    L-\Delta\cdot I\preceq M\preceq L + \Delta\cdot I,
\end{equation}
which in turn implies
\begin{equation}
    (I+L)^{-1}(L-\Delta\cdot I)(I+L)^{-1}\preceq (I+L)^{-1}M(I+L)^{-1}\preceq (I+L)^{-1}(L + \Delta\cdot I)(I+L)^{-1}.
\end{equation}
Finally, note that by Weyl's monotonicity theorem, we have $\max\{\lambda_i(M)-\Delta,0\}\leq \lambda_i(L)\leq \lambda_i(M)+\Delta$. Combining all these bounds with another application of Weyl's monotonicity theorem, we have
\begin{equation}
    \max_{i\in [n]} \frac{\lambda_i(M)-2\Delta}{(1+\lambda_i(M)+\Delta)^2}\leq \lambda_{n}((I+L)^{-1}M(I+L)^{-1}) \leq \max_{i\in [n]} \frac{\lambda_i(M)+2\Delta}{(1+\max\{\lambda_i(M)-\Delta,0\})^2}.
\end{equation}
\end{proof}

That high physical similarity of the graphs implies the problem is not changed significantly is not particularly surprising, though the previous result gives exponentially better dependence on the physical similarity than what can be attained by naive applications of matrix perturbation bounds. However, we now show that high physical similarity edge-by-edge is merely sufficient, but \emph{not necessary}; another relevant property that will ensure that this holds is \emph{spectral similarity}, as defined by Spielman and Teng \cite{spielman2011spectral}.

\begin{defn}
$L$ and $M$ are $\epsilon$-\emph{spectral approximations} for each other for some $\epsilon>0$ if 
\begin{equation*}
    \frac{1}{1+\epsilon}L\preceq M\preceq (1+\epsilon)L.
\end{equation*}
\end{defn}

Note that this definition is symmetric in $L$ and $M$. It is easy to show from this definition that if $L$ and $M$ are $\epsilon$-spectral approximations of each other, then the adversary's objective value cannot differ too much from the single-graph setting with just $M$.
\begin{thm}
\label{thm:spectralsim}
Suppose that $L$ and $M$ are $\epsilon$-spectral approximations of each other. Then
    \begin{equation*}
    \frac{1}{1+\epsilon}\max_{i\in [n]}\min_{c\in [\frac{1}{1+\epsilon}, 1+\epsilon]}\frac{c\lambda_i(M)}{(1+c\lambda_i(M))^2}\leq \lambda_{\text{max}}( (I+L)^{-1}M(I+L)^{-1})\leq (1+\epsilon)\max_{i\in [n]}\max_{c\in [\frac{1}{1+\epsilon}, 1+\epsilon]}\frac{c\lambda_i(M)}{(1+c\lambda_i(M))^2}.
\end{equation*}
\end{thm}
\begin{proof}
The proof is essentially immediate from the definition: pre- and post-multiplying by $(I+L)^{-1}$, we immediately get from the definition that
\begin{equation*}
    \frac{1}{1+\epsilon}(I+L)^{-1}L(I+L)^{-1}\preceq (I+L)^{-1}M(I+L)^{-1}\preceq (1+\epsilon)(I+L)^{-1}L(I+L)^{-1}.
\end{equation*}
By Weyl's monotonicity theorem, this implies the corresponding inequality on each of the eigenvalues. We deduce that
\begin{equation*}
    \frac{1}{1+\epsilon}\lambda_n((I+L)^{-1}L(I+L)^{-1})\preceq \lambda_n((I+L)^{-1}M(I+L)^{-1})\preceq (1+\epsilon)\lambda_n((I+L)^{-1}L(I+L)^{-1}).
\end{equation*}
To relate this back to the matrix $(I+M)^{-1}M(I+M)^{-1}$, we again use Weyl's monotonicity theorem, as then for each $i\in [n]$, 
\begin{equation*}
    \frac{1}{1+\epsilon}\lambda_i(L)\leq \lambda_i(M)\leq (1+\epsilon)\lambda_i(L).
\end{equation*}
In particular, $\lambda_i(L)$ lies in the $(1+\epsilon)$-neighborhood of $\lambda_i(M)$. We showed above that
\begin{equation*}
    \lambda_{\text{max}}( (I+M)^{-1}M(I+M)^{-1})=\max_{i\in [n]}\frac{\lambda_i(M)}{(1+\lambda_i(M))^2};
\end{equation*}
plugging in these ``fuzzy'' versions of the eigenvalues gives the desired inequalities.
\end{proof}

\jgedit{\begin{remark}
Note that while the definition of spectral similarity is symmetric, it need not commute nicely with positive rational expressions of the Laplacians. The reason is that in general, positive rational expressions need not be operator monotone, i.e. may not respect the Loewner order. For instance, $0\preceq A\preceq B$ does not imply $A^2\preceq B^2$, requiring us to appeal to Weyl's monotonicity theorem to translate between $M$ and $L$.
\end{remark}}

As a corollary, this result shows that it is not necessary for $L$ and $M$ to be extremely close in, say, Frobenius or $\ell_1$ norm on each row for the eigenvalues for the adversary's objective value to remain close to the single-graph setting. This is because by seminal results of Batson, Spielman, and Srivastava, \emph{every} graph Laplacian has a weighted $\epsilon$-spectral approximation that corresponds to a graph with $O(n/\epsilon^2)$ edges \cite{batson2012twice}. Necessarily, these graphs are physically quite different, as they can differ in $\Theta(n^2)$ entries. The previous result shows that this is irrelevant; in the mixed-graph objective function, replacing one of these graphs by the other does not meaningfully change the adversary's power to induce disagreement under the Friedkin-Johnsen dynamics.

\subsection{Spectral Dissimilarity}
In this section, we provide a partial converse to the previous section; we provide a simple condition that will imply that the relevant largest eigenvalue is large that relates to the spectral dissimilarity of $L$ and $M$. We then show how this can be realized in the special case of cuts in $G_1$ and $G_2$; it will turn out that if $G_1$ and $G_2$ are highly misaligned in the sense of having even one drastically different vertex cut, then the largest eigenvalue is necessarily large.

\begin{defn}
We say $L$ and $M$ are $(\epsilon,\eta)$-\emph{bad spectral approximations} if there exists $\mathbf{x}\in \mathbb{R}^n$ with $\|\mathbf{x}\|^2=n$ such that
$\mathbf{x}^TL\mathbf{x}\leq \epsilon$
   and $ \mathbf{x}^TM\mathbf{x}\geq \eta.$
\end{defn}

This definition implies that $M$ and $L$ are not $\eta/\epsilon$-spectral approximations for each other, but we will crucially be interested in the actual values, not just the ratio. Moreover, notice that this is not symmetric in the directions of the inequalities, and we will actually care about the numerical values, not just the ratio. For these reasons, the following does not constitute an exact converse, which is essentially immediate from this definition:
\begin{proposition}
\label{prop:badapprox}
Suppose that $L$ and $M$ are $(\epsilon,\eta)$-bad spectral approximations. Then
\begin{equation*}
    \lambda_n((I+L)^{-1}M(I+L)^{-1})\geq \frac{\eta}{n+(\|L\|+2)\epsilon}.
\end{equation*}
\end{proposition}
\begin{proof}
This follows from the variational characterization of eigenvalues:
\begin{align*}
    \lambda_n((I+L)^{-1}M(I+L)^{-1})&=\max_{\mathbf{z}\in \mathbb{R}^n}\frac{\mathbf{z}^T(I+L)^{-1}M(I+L)^{-1}\mathbf{z}}{\mathbf{z}^T\mathbf{z}}\\
    &=\max_{\mathbf{z}\in \mathbb{R}^n}\frac{\mathbf{z}^TM\mathbf{z}}{\mathbf{z}^T(I+L)^2\mathbf{z}}.
\end{align*}
The proof follows from plugging in the guaranteed vector $\mathbf{x}$, simply noting that
\begin{equation*}
    \mathbf{x}^T(I+L)^2\mathbf{x}\leq \mathbf{x}^T(I+2L+L^2)\mathbf{x}\leq n+(\| L\|+2)\epsilon.
\end{equation*}
\end{proof}

This abstract result shows that if $L$ and $M$ are spectrally misaligned in the above sense, then the largest eigenvalue of the mixed-graph objective is large. Tangibly, one specific way that this can occur is the if $L$ and $M$ have very different cut structure. Plugging in characteristic vectors into Proposition \ref{prop:badapprox} and taking the maximum yields
\begin{corollary}
\label{cor:cut}
For any $L,M$,
\begin{align*}
    \lambda_n((I+L)^{-1}M(I+L)^{-1})&\geq \max_{S\subseteq V} \frac{4\mathsf{cut}_{G_2}(S)}{n+4(\|L\|+2)\mathsf{cut}_{G_1}(S)}\\
    &\geq \max_{S\subseteq V} \frac{4\mathsf{cut}_{G_2}(S)}{n+8(\Delta_{G_1}+1)\mathsf{cut}_{G_1}(S)},
\end{align*}
where $\Delta_{G_1}$ is defined to be the largest degree in $G_1$.
\end{corollary}
\begin{proof}
The only new statement comes from noticing $\|L\|\leq 2\Delta_{G_1}$; this follows from the Gershgorin circle theorem, as the maximum absolute row sum of $L$ is at most $2\Delta_{G_1}$.
\end{proof}

 By directly analyzing the Rayleigh quotients, we can also obtain a slightly different bound for such vectors:
\begin{proposition}
\label{prop:cut2}
For any $L,M$, 
\begin{equation*}
    \lambda_n((I+L)^{-1}M(I+L)^{-1})\geq \max_{S\subseteq V} \frac{4\mathsf{cut}_{G_2}(S)/n}{(1+2\sqrt{2}\mathsf{cut}_{G_1}(S)/\sqrt{n})^2}
\end{equation*}
\end{proposition}
\begin{proof}
The proof proceeds analogously by plugging in $\chi_S$ into the Rayleigh quotient for $\lambda_n$. Indeed,
\begin{equation*}
    \lambda_n((I+L)^{-1}M(I+L)^{-1})\geq \frac{\chi_S^TM\chi_S}{\|(I+L)\chi_S\|_2^2}=\frac{4\mathsf{cut}_{G_2}(S)}{\|(I+L)\chi_S\|_2^2}.
\end{equation*}
It suffices to upper bound the denominator. By the Triangle Inequality,
\begin{equation*}
    \|(I+L)\chi_S\|_2\leq \|\chi_S\|_2+\|L\chi_S\|_2=\sqrt{n}+\|L\chi_S\|_2.
\end{equation*}
Moreover, again by the Triangle Inequality
\begin{align*}
    \|L\chi_S\|_2&=\bigg\|\sum_{i\in S,j\not\in S}2w_{G_1}(i,j)(\mathbf{e}_i-\mathbf{e}_j)\bigg\|_2\\
    &\leq 2\sqrt{2} \mathsf{cut}_{G_1}(S).
\end{align*}
Plugging in this estimate and factoring out $n$ gives the claim.
\end{proof}

These results show that if the opinion and disagreement graphs are misaligned on even one large cut of $G_2$, then the adversary will be able to induce disagreement far beyond what is possible in the single-graph objective. As an example, consider an extreme case, where $G_2$ is a complete unweighted bipartite graph on $2n$ nodes, while $G_1$ is two $n$-node cliques on both sides of the bipartition with $o(\sqrt{n})$ edges between them. Then if $S$ is one side of the bipartition, $\mathsf{cut}_{G_2}(S)=n^2$, while $\mathsf{cut}_{G_1}(S)=o(\sqrt{n})$. The estimate given by the Proposition \ref{prop:cut2} yields that the adversary can induce disagreement $\approx 2n$, which is tight even up to constants in light of the upper bound in Lemma \ref{lem:matbound}. This is sharper than the generic bound obtained in Corollary \ref{cor:cut}, which is off asymptotically by a factor of $o(\sqrt{n})$.

\section{Discussion and Open Problems}
In this paper, we have shown how several natural adversarial actions on networks that have been prominent in recent years can be modeled as optimization problems with a standard model of opinion dynamics. By leveraging the well-known connections between the spectral and combinatorial structures of graphs, we are able to gain significant insights into the nature of graphs that are resilient to these outside perturbations.


We believe that studying mixed-graph objectives can be a fruitful future direction; as in our results above, this generalization provides an explanation for how disagreement can arise from outside influence far beyond what can be predicted from the spectral theory arising from considering only a single graph. Similarly, network defense problems like the one we consider here will be of continued relevance as attackers can perturb opinions in a more sophisticated fashion than we have considered here. A more refined study of such problems may give significant actionable insights on how to circumvent these new forms of adversarial behavior.

Two natural, albeit difficult, directions are generalizing these sorts of problems and analyses to directed graphs, as well as considering models where graphs and opinions \emph{co-evolve}, as in the Hegselmann-Krause model \cite{hegselmann2002opinion}. Directed graphs, while having significantly less spectral structures, are a more natural model for influences and opinion dynamics on graphs; similarly, one expects that real-world networks tend to change directly as a result of opinion dynamics. However, even simple models like Hegselmann-Krause pose significant mathematical challenges that make this study quite difficult. Overcoming either of these barriers would enable one to study many more naturally occurring sociological phenomena.

\bibliography{AdversaryBib}
\bibliographystyle{ieeetr}

\end{document}